\documentclass[11pt,reqno]{article}

\usepackage{amssymb,amsmath,amsthm}
\usepackage{latexsym}
 \usepackage{bbm}
 \usepackage[bbgreekl]{mathbbol}
\usepackage{graphicx}
\usepackage{ctable,multirow}
\usepackage[margin=2.5cm]{geometry}
\usepackage{tensor}

\newcommand{\proglang}[1]{\textsf{#1}}
 \newcommand{\bbR}{\mathbb{R}}
 
 \newcommand{\bbX}{\mathbbm{X}}
 \def\argmin{\mathop{\rm arg\,min}}
 
 \def\span{\mathop{\rm span}}
 \def\Null{\mathop{\rm Null}}
  
 \def\bias{\mathop{\rm bias}}
 \def\cov{\mathop{\rm cov}}
 \def\col{\mathop{\rm col}}
 \def\var{\mathop{\rm var}}
 
 \def\diag{\mathop{\rm diag}}
  \def\rank{\mathop{\rm rank}}
   \def\nd{\mathop{\rm o}}
   \def\oo{\mathop{\o}}

 \newcommand{\matr}[2]{\left[ \begin{array}{#1} #2 \end{array} \right]}

\newcommand{\beno}{\begin{eqnarray}}
\newcommand{\eno}{\end{eqnarray}}
\newcommand{\benone}{\begin{eqnarray*}}
\newcommand{\enone}{\end{eqnarray*}}

\newtheorem{theorem}{Theorem}[section]

\newtheorem{proposition}[theorem]{Proposition}
\newtheorem{corollary}[theorem]{Corollary}
\newtheorem{remark}[theorem]{Remark}

\begin{document}

 \setlength{\baselineskip}{16pt}

\title{Structured penalties for functional linear models---partially empirical eigenvectors for regression}

\author{
Timothy W.~Randolph$^{*,\dag}$, Jaroslaw Harezlak$^\#$ and Ziding
Feng$^*$
\\ \bigskip\\
$^*$Fred Hutchinson Cancer Research Center\\
 Biostatistics and Biomathematics\\
 Seattle, WA 98109\\  \\
$^\#$Indiana University School of Medicine \\
Division of Biostatistics\\
Indianapolis, IN 46202
\\ \\
 $^\dag$\texttt{trandolp@fhcrc.org}}

\maketitle

\begin{abstract}
One of the challenges with functional data is incorporating geometric structure, or local
correlation, into the analysis. This structure is inherent in the output from an
increasing number of biomedical technologies, and a functional linear model is often used
to estimate the relationship between the predictor functions and scalar responses. Common
approaches to the problem of estimating a coefficient function typically involve two
stages: regularization and estimation.  Regularization is usually done via dimension
reduction, projecting onto a predefined span of basis functions or a reduced set of
eigenvectors (principal components). In contrast, we present a unified approach that
directly incorporates geometric structure into the estimation process by exploiting the
{\em joint} eigenproperties of the predictors and a linear penalty operator.  In this
sense, the components in the regression are `partially empirical' and the framework is
provided by the generalized singular value decomposition (GSVD). The form of the
penalized estimation is not new, but the GSVD clarifies the process and informs the
choice of penalty by making explicit the joint influence of the penalty and predictors on
the bias, variance and performance of the estimated coefficient function. Laboratory
spectroscopy data and simulations are used to illustrate the concepts.
\end{abstract}


\section{Introduction}

The coefficient function, $\beta$, in a functional linear model (fLM) represents the
linear relationship between responses, $y$, and predictors, $x$, either of which may
appear as a function.   
We consider the special case of scalar-on-function regression, formally written as
$y=\int_I x(t)\beta(t)\, dt+\epsilon$, where $x$ is a random function, square integrable
on a closed interval $I\subset \bbR$, and $\epsilon$ a vector of random i.i.d.~mean-zero
errors. In many instances, one has an approximate idea about the {\em informative
structure} of the predictors, such as the extent to which they are smooth, oscillatory,
peaked, etc. Here we focus on analytical framework for incorporating such information
into the estimation of $\beta$.

The analysis of data in this context involves a set of $n$ responses $\{y_i\}_{i=1}^n$
corresponding to a set of predictor curves $\{x_i\}_{i=1}^n$, each arising as a
discretized sampling of an idealized function; i.e., $x_i \equiv
(x_i(t_1),...,x_i(t_p))$, for some, $t_1<\cdots<t_p$, of $I$.  In particular, the concept
of geometric or spatial structure implies an order relation among the index parameter
values.  We assume the predictor functions have been sampled equally and densely enough
to capture geometric structure of the type typically attributed to functions in
(subspaces of) $L^2(I)$.  For this, it will be assumed that $p>n$ although this condition
is not necessary for our discussion.

Several methods for estimating $\beta$ are based on the eigenfunctions associated with
the auto-covariance operator defined by the predictors \cite{HallHor:07,Muller:05}. These
eigenfunctions provide an empirical basis for representing the estimate and are the basis
for the usual ordinary least-squares and principal-component estimates in multivariate
analysis.  The book by  Ramsay and Silverman \cite{RamSil:05} summarize a variety of
estimation methods that involve some combination of the empirical eigenfunctions and
smoothing, using B-splines or other technique, but none of these methods provide an
analytically tractable way to incorporate presumed structure directly into the estimation
process. The approach presented here achieves this by way of a penalty operator,
$\mathcal{L}$, defined on the space of predictor functions.

The joint influence of the penalty and predictors on the estimated coefficient function
is made explicit by way of the generalized singular value decomposition (GSVD) for a
matrix pair. Just as the ordinary SVD provides the ingredients for an ordinary least
squares estimate (in terms of the empirical basis), the GSVD provides a natural way to
express a penalized least-squares estimate in terms of a basis derived from both the
penalty and the predictors.
We describe this in terms of the $n\times p$ matrix of sampled predictors, $X$, and an
$m\times p$ discretized penalty operator, $L$. The general formulation is familiar as we
consider estimates of $\beta$ that arise from a squared-error loss with quadratic
penalty:
\begin{equation}\label{eq:PenEst}
\tilde\beta_{\alpha,L}=\argmin_\beta\{||y-X\beta||_{\bbR^n}^2 + \alpha||L\beta||_{L^2}^2\}.
\end{equation}
What distinguishes our presentation from others using this formulation is an emphasis on
the {\em joint} spectral properties of the pair $(X,L)$, as arise from the GSVD.  We
investigate the analytical role played by $L$ in imposing structure on the estimate and
focus on how the structure of $L$'s least-dominant singular vectors should be
commensurate with the informative structure of $\beta$.

In a Bayesian view, one may think of $L$ as implementing a prior that favors a
coefficient function lying near a particular subspace; this subspace is determined
jointly by $X$ and $L$. We note, however, that informative priors must come from
somewhere and while they may come from expectations regarding smoothness, other
information often exists---including pilot data, scientific knowledge or laboratory and
instrumental properties. Our presentation aims to elucidate the role of $L$ in providing
a flexible means of implementing informative priors, regardless of their origin.

The general concept of incorporating ``structural information" into regularized
estimation for functional and image data is well established
\cite{BertBocc:98,EnglHankNeub:00,Phillips:62}. Methods for penalized regression have
adopted this by constraining high-dimensional problems in various ``structured" ways
(sometimes with use of an $L^1$ norm): locally-constant structure
\cite{TibFusedLasso:05,Slawski:10}, spatial smoothness \cite{HasBujTib:95},
correlation-based constraints \cite{Tutz:09}, and network-dependence structure described
via a graph \cite{LiLi:08}.  These general penalties have been motivated by a variety of
heuristics: Huang et al.~\cite{Huang:08} refer to the second-difference penalty as an
``intuitive choice"; Hastie et al.~\cite{HasBujTib:95} refer to a ``structured penalty
matrix [which] imposes smoothness with regard to an underlying space, time or frequency
domain"; Tibshirani and Taylor \cite{RTibshTaylor:11} note that the rows of $L$ should
``reflect some believed structure or geometry in the signal"; and the penalties of
Slawski et al.~\cite{Slawski:10} aim to capture ``a priori association structure of the
features in more generality than the fused lasso."

The most common penalty is a (discretized) derivative operator, motivated by the
heuristic of penalizing roughness (see \cite{HasMal:93,RamSil:05}).  Our perspective on
this is more analytical: since the eigenfunctions of the second-derivative operator
$\mathcal{L}=\mathcal{D}^2$ (with zero boundary conditions on $[0,1]$) are of the form
$\varphi(t)=\sin(k\pi t)$, with eigenvalues $k^2\pi^2$ ($k=1,2,..$), $\mathcal{L}$
implements the assumption that the coefficient function is well represented by
low-frequency trigonometric functions. This is in contrast to ridge regression ($L=I$)
which imposes no geometric structure. Although not typically viewed this way, the choice
of $\mathcal{L}=\mathcal{D}^2$, or any differential operator, implies a favored basis for
expansion of the estimate.


A purely empirical basis comprised of a few dominant right singular vectors of $X$ is a
common and theoretically tractable choice. This is the essence of principal component
regression (PCR) and these vectors also form the basis for a ridge estimate. Although
this empirical basis does not technically impose local spatial structure (no order
relation among the index parameter values is used), it may be justified by arguing that a
few principal component vectors capture the ``greatest part" of a set of predictors
\cite{HallPosPre:01}. Properties of this approach for signal regression is the focus of
\cite{CaiHall:06} and \cite{HallHor:07}. The functional data analysis framework of Ramsay
and Silverman \cite{RamSil:05} provides two formulations of PCR. One in which the
predictor curves are themselves smoothed prior to construction of principal components
(chap.~8) and another that incorporates a roughness penalty into the construction of
principal components (chap.~9), as originally proposed in \cite{Silverman:96}.  In a
related presentation on signal regression, Marx and Eilers \cite{MarxEilers:99} proposed
a penalized B-spline approach in which predictors are transformed using a basis external
to the problem (B-splines) and the estimated coefficient function is derived using the
transform coefficients.  Combining ideas from \cite{MarxEilers:99} and \cite{HasMal:93},
the ``smooth principal components" method of \cite{CarFerSar:03} projects predictors onto
the dominant eigenfunctions to obtain an estimate then uses B-splines in a procedure that
smooths the estimate. Reiss and Ogden \cite{ReissOgden:07} provide a thorough study on
several of these methods and propose modifications that include two versions of PCR using
B-splines and second-derivative penalties: FPCR$_{C}$ applies the penalty to the
construction of the principal components (cf.~\cite{Silverman:96}), while FPCR$_{R}$
incorporates the penalty into the regression (cf.~\cite{RamSil:05}).

In the context of nonparametric regression ($X=I$) the formulation \eqref{eq:PenEst}
plays a dominant role for smoothing \cite{WahbaBook:90}. Related to this, Heckman and
Ramsay \cite{HeckRams:00} proposed a differential equations model-based estimate of a
function $\mu$ whose properties are determined by a linear differential operator chosen
from a parameterized family of differential equations, $L\mu=0$. In this  context,
however, the GSVD is irrelevant since $X$ does not appear and the role of $L$ is
relatively transparent.

Algebraic details on the GSVD as it relates to penalized least-squares are given in
section~\ref{sec:PenRegGSVD} with analytic expressions for various properties of the
estimation process are described in section~\ref{sec:BiasVar}. Intuitively, smaller bias
is obtained by an informed choice of $L$ (the goal being small $L\beta$). The affect of
such a choice on the variance is described analytically.
Section~\ref{sec:StructuredPenSec} describes several classes of structured penalties
including two previously-proposed special cases that were justified by numerical
simulations. The targeted penalties of subsection \ref{sec:TargetedPen} are studied in
more detail in section~\ref{sec:AnalProps} including an analysis of the mean squared
error for a family of penalized estimates which encompasses the ridge,
principal-component and James-Stein estimates.

The assumptions on $L$ here are increasingly restrictive to the point where the estimates
are only minor extensions of these well-studied estimates. The goal, however, is to
analytically describe the substantial gains achievable by even mild extensions of these
established methods.

In applications the selection of the tuning parameter, $\alpha$ in \eqref{eq:PenEst}, is
important and so Section~\ref{sec:TuningParam} describes our application of REML-based
estimation for this. Numerical illustrations are provided in section~\ref{sec:NumProps}:
the simulation in subsection~\ref{sec:BumpsSim} is motivated by Reiss and Ogden's study
of fLMs \cite{ReissOgden:07}; \ref{sec:RamanSim} presents a simulation using
experimentally-derived Raman spectroscopy curves in which the ``true" $\beta$ has
naturally-occurring (laboratory) structure; and section~\ref{sec:RamanApplic} presents an
application based on experimentally collected spectroscopy curves representing varied
biochemical (nanoparticle) concentrations.  An appendix  looks at the simulation studied
by Hall and Horowitz \cite{HallHor:07}.  We begin in section~\ref{sec:BkgdNotn} with a
brief setup for notation and an introductory example. Note that for any $L\ne I$, the
estimated $\beta$ is not given in terms of the ordinary empirical singular vectors (of
$X$), but rather in terms of a ``partially empirical" basis arising from a simultaneous
diagonalization of $X'X$ and $L'L$ via the GSVD. Hence, for brevity, we refer to
$\tilde\beta_{\alpha,L}$ as a PEER (partially empirical eigenvector for regression)
estimate whenever $L\ne I$.


\section{Background and simple example}\label{sec:BkgdNotn}

Let $\beta$ represent a linear functional on $L^2(I)$ defining a linear relationship $y=
\int_I x(t)\beta(t)\, dt +\epsilon$ (observed with error, $\epsilon$) between a response,
$y$, and random predictor function, $x\in L^2(I)$. We assume a set of $n$ scalar
responses $\{y_i\}_{i=1}^n$ corresponding to the set of $n$ predictors,
$\{x_i\}_{i=1}^n$, each discretely sampled at common locations in $I$.  Denote by $X$ the
$n\times p$ matrix whose $i$th row is a $p$-dimensional vector, $x_i$, of discretely
sampled functions, and columns that are centered to have mean $0$.  The notation $\langle
\cdot, \cdot \rangle $ will be used to denote the inner product on either $L^2(I)$ or
$\bbR^p$, depending on the context.

The empirical covariance operator is ${K}=\frac{1}{n}X'X$, but for functional predictors,
typically $p>n$ or else $K$ is ill-conditioned or rank deficient.  In this case, there
are either infinitely many least-squares solutions, $\hat\beta\equiv
\argmin_{\beta}||y-X\beta||^2$,  or else any such solution is highly unstable and of
little use.  The least-squares solution having minimum norm is unique, however, and it
can be obtained directly by the singular value decomposition (SVD): $X =UDV'$ where the
left and right singular vectors, $u_k$ and $v_k$, are the columns of $U$ and $V$,
respectively, and $D=[D_1 \ 0]$, where $D_1=\mbox{diag}\{\sigma_k\}_{k=1}^n$, typically
ordered as $\sigma_1 \ge \sigma_2\ge ...\ge \sigma_r>0$ ($r=\mbox{rank(X)}$).
In terms of the SVD of $X$, the minimum-norm solution is $\hat\beta_+ =  X^\dag y =
\sum_{\sigma_k\ne 0} ({1}/{\sigma_k}) u_k'y\,v_k$, where $X^\dag$ denotes the
Moore-Penrose inverse of $X$: $X^\dag=VD^\dag U'$, where $D^\dag=\mbox{diag}\{1/\sigma_k
\mbox{ if } \sigma_k\ne 0 ; \ 0 \mbox{ if }\sigma_k= 0\}$.
The orthogonal vectors that form the columns of $V$ are the eigenvectors of $X'X$ and are
sometimes referred to as a Karhunen-Lo\`eve (K-L) basis for the row space of $X$.

The solution $\hat\beta_+$ is Marquardt's generalized inverse estimator whose properties
are discussed in \cite{Marq:70}.  For functional data, $\hat\beta_+$ is an unstable,
meaningless solution. One obvious fix is to truncate the sum to $d<r$ terms so that
$\{\sigma_k\}_{k=1}^{d}$ is bounded away from zero. This leads to the truncated singular
value or principal component regression (PCR) estimate: $\tilde\beta_{\mbox{\tiny
PCR}}\equiv\tilde\beta^{d}_{\mbox{\tiny PCR}}=V_{d} D_{d}^{-1}{U_{d}}'y \, $ where here,
and subsequently, we use the notation $A_{d}\equiv\col[a_1,...,a_{d}]$ to denote the
first $d$ columns of a matrix $A$.

When $L=I$, the minimizer in \eqref{eq:PenEst} is the ridge penalty due to A.~E.~Hoerl
\cite{HoerlKen:70}
\begin{equation}\label{eq:betaRidge}
\tilde\beta_{\alpha,I} = (X'X+\alpha I)^{-1}X'y =
\sum_{k=1}^r \left(\frac{\sigma_k^2}{\sigma_k^2+\alpha}\right)\,\frac{1}{\sigma_k} u_k'y\,v_k,
\end{equation}
or, $\tilde\beta_{\alpha,I}=V_rF^{\alpha}D_r^\dag {U_r}'y$, where
$F^\alpha=\diag\{\frac{\sigma_k^2}{\sigma_k^2+\alpha}\}$.  The factor $F^\alpha$  acts to
counterweight, rather than truncate, the terms $\frac{1}{\sigma_k}$ as they get large.
This is one of many possible filter factors which address problems of ill-determined rank
(for more, see \cite{EnglHankNeub:00,HansenBook:98,Neum:98}). Weighted (or generalized)
ridge regression replaces $L=I$ with a diagonal matrix whose entries downweight those
terms corresponding to the most variation \cite{HoerlKen:70}.  Other ``generalized ridge"
estimates replace $L=I$ by a discretized second-derivative operator, $L=\mathcal{D}^2$.
Indeed, the Tikhonov-Phillips form of regularization \eqref{eq:PenEst} has a long history
in the context of differential equations \cite{Tikhonov:43,Phillips:62} and image
analysis \cite{Groetsch:84,Neum:98} with emphasis on numerical stability.  In a linear
model context, the smoothing imposed by this penalty was mentioned by Hastie and Mallows
\cite{HasMal:93}, discussed in Ramsay and Silverman \cite{RamSil:05} and used (on a the
space of spline-transform coefficients) by Marx and Eilers \cite{MarxEilersCalib:02},
among others.
The following simple example illustrates basic behavior for some of these penalties
alongside an idealized PEER penalty.

\subsection{A simple example} \label{sec:BumpsIntro}
We consider a set of $n=50$ bumpy predictor curves \{$x_i$\} discretely sampled at
$p=250$ locations, as displayed in gray in the last panel of Figure~\ref{Fig:PartSums}.
The true coefficient function, $\beta$, is displayed in black in this same panel. The
responses are defined as $y_i =\langle x_i, \beta \rangle + \epsilon_i$  ($\epsilon_i$
normal, uncorrelated mean-zero errors), and hence depend on the amplitudes of $\beta$'s
three bumps centered at locations $t=45,110,210$.

A detailed simulation with complete results are provided in section~\ref{sec:BumpsSim}.
Here we simply illustrate the estimation process for $L=I$, as in \eqref{eq:betaRidge},
in comparison with $L=\mathcal{D}^2$ and an idealized PEER penalty.  The latter is
constructed using a visual inspection of the predictors and lightly penalizes the
subspace spanned by such structure, specifically, bumps centered at all visible locations
(approximately $t=15, 45, 80, 110, 160, 210, 240$).

The first five panels serve to emphasize the role played by the structure of basis
vectors that comprise the series expansion in \eqref{eq:betaRidge} (in terms of ordinary
singular vectors) versus the analogous expansion (see \eqref{eq:betaGSVE}) in terms of
generalized singular vectors.  In particular, Figure \ref{Fig:PartSums} shows several
partial sums of \eqref{eq:betaGSVE} for these three penalties. The ridge process (gray)
is, naturally, dominated by the right singular vectors of $X$ which become increasingly
noisy in successive partial sums. The second-derivative penalized estimate (dashed) is
dominated by low-frequency structure, while the targeted PEER estimate converges quickly
to the informative features.

\begin{figure}
\centering
{\includegraphics[height=6cm,width=11cm]{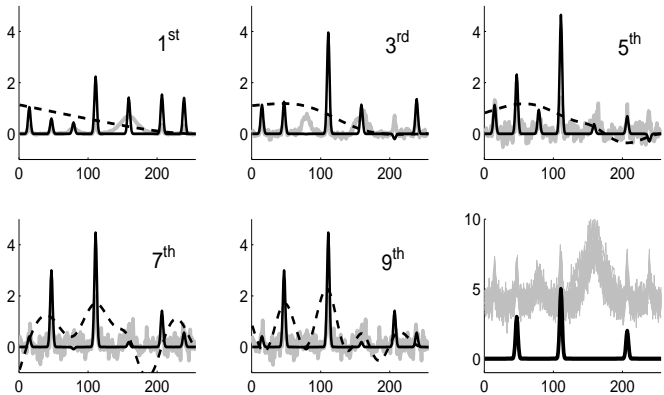}}
\vskip-10pt\caption{\small {\bf Partial sums of penalized estimates.}
The first five odd-numbered partial sums from \eqref{eq:betaGSVE} for three penalties:
2nd-derivative (dashed), ridge (gray), targeted PEER (black; see text in
sections~\ref{sec:BumpsIntro} and \ref{sec:BumpsSim}).
The last panel displays $\beta$ (black) and 15 predictors, $x_i$ (gray),
from the simulation.}
 \label{Fig:PartSums}
\end{figure}

In this toy example, visual structure (spatial location) is used to define a
regularization process that easily outperforms uninformed methods of penalization. Less
visual examples where the penalty is defined by a set of laboratory-derived structure (in
Raman spectroscopy curves) is given in sections~\ref{sec:RamanSim} and
\ref{sec:RamanApplic}; see Figure~\ref{Fig:COINs}. In that setting, and in general, the
role played by $L$ is appropriately viewed in terms of a preferred subspace in $\bbR^p$
determined by its singular vectors. Algebraic details about how structure in the
estimation process is determined jointly by $X$ and $L\ne I$  are described next.

\section{Penalized least squares and the GSVD}\label{sec:PenRegGSVD}

Of the many methods for estimating a coefficient function discussed in the Introduction,
nearly are all aimed at imposing geometric or ``functional" structure into the process
via the use of basis functions in some manner.  An alternative to choosing a basis
outright is to exploit the structure imposed by an informed choice of penalty operator.
The basis, determined by a pair $(X,L)$, can be tailored toward structure of interest by
the choice of $L$.  When this is carried out in the least-squares setting of
\eqref{eq:PenEst}, the algebraic properties of the GSVD explicitly reveal how the
structure of the estimate is inherited from the spectral properties of $(X,L)$.

\subsection{The GSVD}
For a given linear penalty $L$ and parameter $\alpha>0$, the estimate in
\eqref{eq:PenEst} takes the form
\begin{equation}\label{eq:beta_NormalEqns}
\tilde{\beta}_{\alpha,L} = (X'X +\alpha L'L)^{-1}X'y.
\end{equation}
This cannot be expressed using the singular vectors of $X$ alone, but the generalized
singular value decomposition of the pair $(X,L)$ provides a tractable and interpretable
series expansion. The GSVD appears in the literature in a variety of forms and notational
conventions.  Here we provide the necessary notation and properties of the GSVD for our
purposes (see, e.g., \cite{HansenBook:98}) but refer to
\cite{Bjorck:96,GolubVanLoan:96,PaigSaun:81} for a complete discussion and details about
its computation.  See also the comments of Bingham and Larntz \cite{BingLarntz:77}.

Assume $X$ is an $n\times p$ matrix ($n\le p$) of rank $n$, $L$ is an $m\times p$ matrix
($m\le p$) of rank $m$. We also assume that $n \le m \le p \le m+n$, and the rank of the
$(n+m)\times p$ matrix $Z:=[X'\ L']'$ is $p$. A unique solution is guaranteed if the null
spaces of $X$ and $L$ intersect trivially: $\Null(L)\cap\Null(X)=\{0\}$. This is not
necessary for implementation, but it is natural in our applications and simplifies the
notation. In addition, the condition $p>n$ is not required, but rather than present
notation for multiple cases, this will be assumed.

Given $X$ and $L$, the following matrices exist and form the decomposition below: an
$n\times n$ matrix $U$ and an $m\times m$ matrix $V$, each with orthonormal columns,
$U'U=I$, $V'V=I$;  diagonal matrices $S$ ($n\times n$) and $M$ ($m\times m$); and a
nonsingular $p\times p$ matrix $W$ such that
\begin{equation}\label{eq:GSVD}
\begin{aligned}
X & =U\underline{S}W^{-1}\, , \qquad \ \underline{S}=\matr{cc}{0 & S},\qquad S=\text{diag}\{\sigma_k\}\\
L & =V\underline{M}W^{-1}\,  \qquad \underline{M}=\matr{cc}{ M &  0}, \qquad M=\text{diag}\{\mu_k\}.
\end{aligned}
\end{equation}
Here,  $S$ and $M$ are of the form    $S =\matr{cc}{S_1 & 0\\ 0 & I_{p-m} }$ and
$M=\matr{cc}{I_{p-n} & 0 \\ 0 & M_1 }$, whose submatrices $S_1$ and $M_1$ have
 $q:= n+m-p$ diagonal entries ordered as
\begin{equation}\label{eq:GS_values}
\begin{aligned}
0 & \le \sigma_1 \le \sigma_2 \le \cdots \le \sigma_q\le 1 \\
1 & \ge \mu_1 \ge \mu_2 \ge \cdots \ge \mu_q\ge  0
\end{aligned}
\qquad \mbox{where} \qquad \sigma_k^2+\mu_k^2=1, \quad k=1,...,q,
\end{equation}
These matrices satisfy
\begin{equation}\label{eq:W'X'XW}
W'X'XW =  \matr{ccc}{0 & 0 & 0 \\ 0 & S_1^2 & 0 \\ 0 & 0 & I} = \underline{S}'\underline{S}, \qquad
W'L'LW = \matr{ccc}{I & 0 & 0 \\ 0 & M_1 ^2 & 0\\  0 & 0 & 0} = \underline{M}'\underline{M},
\end{equation}
with $\underline{S}'\underline{S}+\underline{M}'\underline{M}=I$.

Denote the columns of $U$, $V$ and $W$ by $u_k$, $v_k$ and $w_k$, respectively. In spite
of the notation, the generalized singular vectors $u_k$ and $v_k$ are not the same as the
ordinary singular vectors of $X$ in Section~\ref{sec:BkgdNotn}.  They are the same when
$L=I$, although their ordering is reversed; in that case, the ordinary singular values
correspond to $\gamma_k:=\sigma_k/\mu_k$ for $\mu_k>0$. By the convention used for
ordering the GS values and vectors, the last few columns of $W$ span the subspace
$\Null(L)$ (or, if $\Null(L)$ is empty, they correspond to the smallest GS values,
$\mu_k$). We set $d=\dim(\Null(L))$ and note that $\mu_k=0$ for $k>n-d$.

Now, equation \eqref{eq:W'X'XW} and some algebra gives $(X'X+\alpha L'L)^{-1}
 = W (\underline{S}'\underline{S} + \alpha\underline{M}'\underline{M})^{-1} W'$, and so
 $\tilde\beta_{\alpha,L}
 = W(\underline{S}'\underline{S}+\alpha \underline{M}'\underline{M})^{-1}\underline{S}'U'y$.
A consequence of the ordering adopted for the GS values and vectors is that the first
$p-n$ columns of $W$ play no role in this solution; see equation~\eqref{eq:GSVD}.  So we
can replace $W$ by the $p\times n$ matrix, $W_n$, consisting of the last $n$ columns of
$W$ (corresponding to the indices $p-n+1$ to $p$).  We index the columns of $W_n$ as
$w_1$,...,$w_n$ which is consistent with the indexing established in \eqref{eq:GS_values}
for the singular values in $S$ and $M$. Therefore, the $L$-penalized estimate can be
expressed as a series in terms of GS values as
\begin{equation}\label{eq:betaGSVE}
\tilde\beta_{\alpha,L}
 = \sum_{k=1}^{n-d} \left(\frac{\sigma_k^2}{\sigma_k^2+\alpha \mu_k^2}\right) \frac{1}\sigma_k u_k'y \, w_k
 \, + \sum_{k=n-d+1}^n  u_k'y \, w_k.
\end{equation}
This GSV expansion  corresponds to a new basis for the estimation process: the estimate
is expressed in terms of GS vectors $\{w_k\}$ determined jointly by $X$ and $L$; cf.~the
ridge estimate in~\eqref{eq:betaRidge}.

For more compact notation used later, define the $(n-d)\times (n-d)$ diagonal matrix
$\Gamma =\diag\{\gamma_k = \sigma_k/\mu_k\}_{k=1}^{n-d}$.  For brevity, set $\nd:=n-d$
and let $A_{\nd}$ denote the first $\nd$ columns of a matrix $A$. Also, denote by
$A_{\oo}$ the {\em last} $d$ columns of $A$. In particular, the range of $W_{\oo}$ is
$\Null(L)$. Using this notation, \eqref{eq:betaGSVE} may be written concisely as
\begin{equation}\label{eq:beta_aLQ}
 \tilde{\beta}_{\alpha,L}=  W_{n \nd} F^{\alpha}\Gamma^\dag {U_{\nd}}'y + W_{\oo}{U_{\oo}}'y,
\end{equation}
where $F^\alpha=\diag\left\{\frac{\sigma_k^2}{\sigma_k^2+\alpha\mu_k^2}\right\}_{k =
1}^{n-d}$.

In summary, the utility of a penalty $L$ depends on whether the true coefficient function
shares structural properties with this GSVD basis, $\{w_k\}_{k=1}^n$. With regard to
this, the importance of the parameter $\alpha$ may be reduced by a judicious choice of
$L$ since the terms in \eqref{eq:betaGSVE} corresponding to the vectors $\{w_k:\mu_k=0\}$
are independent of the parameter $\alpha$ \cite{Varah:79}.

As we'll see, bias enters the estimate to the extent that the vectors $\{w_k:\mu_k\ne
0\}$ appear in the expansion \eqref{eq:betaGSVE}. The portion of
$\tilde{\beta}_{\alpha,L}$ that extends beyond the subspace $\Null(L)$ is constrained by
a sphere (of radius determined by $\alpha$);  this portion corresponds to bias. Hence,
$L$ may be chosen in such a way that the bias and variance of $\tilde\beta_{\alpha,L}$
arises from a specific type of structure, potentially decreasing bias without increasing
complexity of the model.  As a common example, $L=\mathcal{D}^2$ introduces smooth bias
structured by low-frequency trigonometric functions.



\subsection{Bias and variance and the choice of penalty operator} \label{sec:BiasVar}
Begin by observing that the penalized estimate  $\tilde\beta_{\alpha,L}$ in
\eqref{eq:beta_NormalEqns} is a linear transformation of any solution to the
normal equations. Indeed, define
 $ X^\# \equiv X^\#_{\alpha,L} = (X'X +\alpha L'L)^{-1}X'$
and note that if $\hat\beta$ denotes any solution to $X'X\beta=X'y$, then
$\tilde\beta_{\alpha,L} = X^\#X\hat\beta + X^\#\epsilon$.  The {\em resolution} operator
$X^\#X$ reflects the extent to which the estimate in \eqref{eq:betaGSVE} is linearly
transformed relative to an exact solution.  In particular, $E(\tilde{\beta}_{\alpha,L})=
X^\#X\hat\beta$. Additionally, we have $\mbox{bias}(\tilde{\beta}_{\alpha,L})=(I-X^\#
X)\beta=\alpha(X'X +\alpha L'L)^{-1}L'L\beta$, and so
$||\mbox{bias}(\tilde{\beta}_{\alpha,L})|| \le ||\alpha(X'X +\alpha L'L)^{-1}L'||\,
||L\beta||$. Hence bias can be controlled by the choice of $L$, with an estimate being
unbiased whenever $L\beta=0$.  There is a tradeoff, of course, and
equation~\eqref{eq:var} below quantifies the effect on the variance as determined by
$W_{\oo}$ (i.e., $\{w_k\}_{k=n-d+1}^n$) if $\Null(L)$ is chosen to be too large.

More generally, the decompositions in~\eqref{eq:GSVD} lead to an expression for the
resolution matrix as  $X^{\#}X = W(\underline{S}'\underline{S} + \alpha
\underline{M}'\underline{M})^{-1}\underline{S}'\underline{S}W^{-1}$, and so $I - X^{\#}X
= \alpha W(\underline{S}'\underline{S} + \alpha
\underline{M}'\underline{M})^{-1}\underline{M}'\underline{M}W^{-1}$. Again, from
equation~\eqref{eq:GSVD}, the first $p-n$ rows of $W^{-1}$ are not used.  For notational
convenience, define $\tilde{W}:=W'^{-1}$  (note, $\tilde{W}$ plays a role analogous to
$V\equiv V'^{-1}$ in the SVD). As before, let $\tilde{W}_n$ denote the $p\times n$ matrix
consisting of the last $n$ columns of $\tilde{W}$ and note that in
equation~\eqref{eq:GSVD}, $X=U\underline{S}W^{-1} = US(\tilde{W}_n)'$.  Hence, $I-X^{\#}X
= \alpha W_n(S^2 + \alpha M^2)^{-1}M^2(\tilde{W}_n)'$, and so the bias of
$\tilde{\beta}_{\alpha,L}$ can be expressed as
\begin{equation}\label{eq:bias}
\mbox{bias}(\tilde{\beta}_{\alpha,L})
   = (I-X^\# X)\beta
  = \sum_{k=1}^{n-d} \frac{\alpha \mu_k^2}{\sigma_k^2+\alpha\mu_k^2}w_k \tilde{w_k}'\beta
\end{equation}
where $\tilde{w}_k$ is the $k$th column of  $\tilde{W}$.  In particular, the bias vector
\eqref{eq:bias} is contained in $\span\{w_k:\mu_k\ne0\}$, whereas the estimate
$\tilde{\beta}_{\alpha,L}$ is in $\span\{w_k:\sigma_k\ne0\}$. In the special case that
$X'X$ is invertible, then $\beta=WS^{-1} U'y$ and
$\mbox{bias}(\tilde{\beta}_{\alpha,L})=\sum_{k} \left(\frac{\alpha
\mu_k^2}{\sigma_k^2+\alpha\mu_k^2}\right)\frac{1}{\sigma_k}u_k'y\, w_k$ (see
\cite{MacLeod:88}).

A counterpart is an expression for the variance in terms of the GSVD. Let $\Sigma$ denote
the covariance for $\epsilon$. Then
$\var(\tilde\beta_{\alpha,L})  
= \var(X^{\#}X\beta + X^{\#}\epsilon) = X^{\#}\Sigma(X^{\#})'$.  Assuming
$\Sigma=\sigma_{\epsilon} I$, this simplifies to
\begin{equation}\label{eq:var}
\var(\tilde\beta_{\alpha,L}) = \sigma_{\epsilon}^2 X^{\#}(X^{\#})'
 = \sigma_{\epsilon}^2 \left( \sum_{k=1}^{n-d} \frac{\sigma_k^2}{(\sigma_k^2+\alpha\mu_k^2)^2}w_kw_k'
 + \sum_{k=n-d+1}^{n}w_kw_k'\right).
\end{equation}

An interesting perspective of the bias-variance tradeoff
is provided by the relationship between the GS-values in \eqref{eq:GS_values}
and their role in equations \eqref{eq:bias} and \eqref{eq:var}. Moreover, these
lead to an explicit expression for the mean squared error (MSE) of a PEER
estimate. Since $E(\tilde\beta_{\alpha,L})=X^\#X\beta$,
\begin{equation}\label{eq:MSE_beta}
\begin{aligned}
\mbox{MSE}(\tilde\beta_{\alpha,L}) & = E(||\beta-\tilde\beta_{\alpha,L}||^2)
    = E(||\beta||^2 + ||\tilde\beta_{\alpha,L}||^2 - 2\langle\beta, \tilde\beta_{\alpha,L} \rangle) \\
   & = ||\beta-X^\#X\beta||^2 + \sigma_\epsilon^2\,\mbox{trace}(X^\#{X^\#}')  \\
   & = ||\bias(\tilde\beta_{\alpha,L})||^2 +
       \sigma_{\epsilon}^2 \sum_{k=1}^{n-d} \frac{\sigma_k^2}{(\sigma_k^2+\alpha\mu_k^2)^2}||w_k||^2.
\end{aligned}
\end{equation}
The GS-vectors $\{w_k\}$ are not necessarily orthogonal, although they are
$X'X$-orthogonal; see \eqref{eq:W'X'XW}. Consequently, a bound, rather than equality, for
the MSE in terms of the GS values/vectors is the best one can do in general:
\begin{equation}
 \mbox{MSE}(\tilde\beta_{\alpha,L})  \le
  \left( \sum_{k=1}^{n-d} \frac{\alpha\mu_k^2}{\sigma_k^2+\alpha\mu_k^2} \tilde{w_k}'\beta\, ||w_k||\right)^2
   + (\mbox{the second term in \eqref{eq:MSE_beta}}).
\end{equation}

As a final remark, recall that one perspective on ridge estimation defines fictitious
data from an orthogonal ``experiment," represented by an $L$, and expresses $I$ as
$I=L'L$ \cite{Marq:70}.  Regardless of orthogonality this applies to any penalized
estimate and $L$ may similarly be viewed as augmenting the data, influencing the
estimation process through its eigenstructure; the response, $y$, is set to zero for
these supplementary ``data".  In this view, equation~\eqref{eq:beta_NormalEqns} can be
written as $Z\beta=
\underline{y}$ where $Z=\matr{c}{X\\ \sqrt{\alpha} L}$ and $\underline{y}=\matr{c}{y\\
0}$. This formulation proves useful in section~\ref{sec:MSEcompare} when assuring that
the estimation process is stable with respect to perturbations in $X$ and the choice of
$L$.

\section{Structured penalties}\label{sec:StructuredPenSec}

A {\em structured penalty} refers to a second term in \eqref{eq:PenEst} that
involves an operator chosen to encourage certain functional properties in the
estimate. A prototypical example is a derivative operator which imposes
smoothness via its eigenstructure.  Here we describe several examples of
structured penalties, including two that were motivated heuristically and
implemented without regard to the spectral properties that define their
performance. Sections~\ref{sec:BiasVar} and \ref{sec:MSEcompare} provide a
complete formulation of their properties as revealed by the GSVD.

\subsection{The penalty of C. Goutis}\label{sec:GoutisPen}

The concept of using a penalty operator whose eigenstructure is targeted toward
specific properties in the predictors appears implicitly in the work of C.
Goutis \cite{Goutis:98}.  This method aimed to account for the ``functional
nature of the predictors" without oversmoothing and, in essence, considered the
inverse of a smoothing penalty.  Specifically, if $\Delta$ denotes a
discretized second-derivative operator (with some specified boundary
conditions), the minimization in \eqref{eq:PenEst} was replaced by
$\min_\beta\{||Y-X\Delta'\Delta\beta||_{\bbR^n}^2 +
\alpha||\Delta\beta||_{L^2}^2\}$. Here, the term $X\Delta'\Delta\beta$ can be
viewed as the product of $X\Delta'$ (derivatives of the predictor curves) and
$\Delta\beta$ (derivative of $\beta$). Defining $\gamma:=\Delta'\Delta\beta$
and seeking a penalized estimate of $\gamma$ leads to
\begin{equation}\label{eq:GoutisPenEq}
\begin{aligned}
\tilde{\gamma} & = (X'X+\alpha(\Delta'\Delta)^\dag)^{-1}X'y \\
& = \argmin_\gamma\{||y-X\gamma||^2+\alpha \langle\gamma,(\Delta'\Delta)^\dag\gamma\rangle\}.
\end{aligned}
\end{equation}
In \cite{Goutis:98}, the properties of $\tilde{\gamma}$ were conjectured to
result from the eigenproperties of $(\Delta'\Delta)^\dag$. This was explored by
ignoring $X$ and plotting some eigenvectors of $(\Delta'\Delta)^\dag$.  The
properties of this method become transparent, however, when formulated in terms
of the GSVD.  That is, let $L:=((\Delta'\Delta)^\dag)^{1/2}$ and note the
functions that define $\hat\gamma$ are influenced most by the highly
oscillatory eigenvectors of $L$ which correspond to its smallest eigenvalues;
see  equations \eqref{eq:GS_values} and \eqref{eq:betaGSVE}.

This approach was applied in \cite{Goutis:98} only for prediction and has drawbacks in
producing an interpretable estimate, especially for non-smooth predictor curves. The
general insight is valid, however, and modifications of this penalty can be used to
produce more stable results.  The operator $(\Delta'\Delta)^\dag$  essentially reverses
the frequency properties of the eigenvectors of $\Delta$ and is an extreme alternative to
this smoothing penalty. An eigenanalysis of the pair $(X,L)$, however, suggests penalties
that may be more suited to the problem. This is illustrated in
Section~\ref{sec:NumProps}.

\subsection{Targeted penalties}\label{sec:TargetedPen}

Given some knowledge about the relevant structure, a penalty can be defined in terms of a
subspace containing this structure. For example, suppose
$\beta\in\mathcal{Q}:=\span\{q_j\}_{j=1}^d$ in $L^2(I)$.  Set $Q=\sum_{j=1}^d q_j\otimes
q_j$ and consider the orthogonal projection $P_\mathcal{Q}=QQ^\dag$. (Here, $q\otimes q$
denotes the rank-one operator $f\mapsto \langle f,q\rangle q$, for $f\in L^2(I)$.) For
$L=I-P_\mathcal{Q}$, then $\beta\in\Null(L)$ and $\tilde{\beta}_{\alpha,L}$ is unbiased.
The problem may still be underdetermined so, more pragmatically, define a
decomposition-based penalty
\begin{equation}\label{eq:LQPQ}
L \equiv L_\mathcal{Q}=a(I-P_\mathcal{Q})+b P_\mathcal{Q}
\end{equation}
for some $a,b\ge 0$. Heuristically, when $a > b > 0$ the effect is to move the estimate
towards $Q$ by preferentially penalizing components orthogonal to $\mathcal{Q}$; i.e.,
assign a prior favoring structure contained in the subspace $\mathcal{Q}$. To implement
the tradeoff between the two subspaces, we view $a$ and $b$ as inversely related, $ab=$
const. The analytical properties of estimates that arise from this are developed in the
next section and illustrated numerically in Section~\ref{sec:NumProps}. For example, bias
is substantially reduced when $\beta\subset\mathcal{Q}$, and
equation~\eqref{eq:MSE_beta_abV} quantifies the tradeoff with respect to variance when
the prior $\mathcal{Q}$ is chosen poorly.

More generally, one may penalize each subspace differently by defining $L=\alpha_1
(I-P_\mathcal{Q})\mathcal{L}_1(I-P_\mathcal{Q}) + \alpha_2 P_\mathcal{Q} \mathcal{L}_2
P_\mathcal{Q}$, for some operators $\mathcal{L}_1$ and $\mathcal{L}_2$. This idea could
be carried further: for any orthogonal decomposition of $L^2(I)$ by subspaces
$\mathcal{Q}_1,\dots,\mathcal{Q}_J$, let $P_j$ be the projection onto $\mathcal{Q}_j$.
Then the multi-space penalty $L=\sum_{j=1}^J \alpha_jP_j$ leads to the estimate
\begin{equation*}\label{eq:betahat_Wj}
\tilde\beta=\argmin_\beta \{ ||y-X\beta||^2 + \sum_{j=1}^J \alpha_j||P_j\beta||^2\}.
\end{equation*}
This concept was applied in the context of image recovery (where   $X$ represents a
linear distortion model for a degraded image $y$) by Belge et al.~\cite{Belge:00}.

The examples here illustrate ways in which assumptions about the structure of a
coefficient function can be incorporated directly into the estimation process. In
general, any estimation of $\beta$ imposes assumptions about its structure (either
implicitly or explicitly) and section~\ref{sec:BiasVar} shows that the bias-variance
tradeoff involves a choice on the {\em type} of bias (spatial structure) as well as the
{\em extent} of bias (regularization parameter(s)).

\section{Some analytical properties}\label{sec:AnalProps}

Any direct comparison between estimates using different penalty operators is confounded
by the fact there is no simple connection between the generalized singular values/vectors
and the ordinary singular values/vectors.  Therefore, we first consider the case of
targeted or projection-based penalties \eqref{eq:LQPQ}.  Within this class, we introduce
a parameterized family of estimates that are comprised of {\em ordinary} singular
values/vectors.  Since the ridge and PCR estimates are contained in (or a limit of) this
family, a comparison with some targeted PEER estimates is possible. For more general
penalized estimates, properties of perturbations provide some less precise relationships;
see proposition~\ref{prop:MSE_perturb}.

%

\subsection{Transformation to standard form}\label{sec:StdForm}

We have reason to consider decomposition-based penalties \eqref{eq:LQPQ} in which $L$ is
invertible.  In this case, an expression for the estimate does not involve the second
term in \eqref{eq:beta_aLQ}, and decomposing the first term into two parts will be
useful. For this, we find it convenient to use the standard-form transformation due to
Elden \cite{Elden:82} in which the penalty $L$ is absorbed into $X$.  This transformation
also provides a computational alternative to the GSVD which, for projection-based
penalties in particular, can be less computationally expensive; see, e.g.,
\cite{KilHanEsp:07}. By this transformation of $X$, a general PEER estimate ($L\ne I$)
can be expressed via a ridge-regression process.

Define the {\em $X$-weighted generalized inverse of $L$} and the corresponding
transformed $X$ as:
\begin{equation*}
L_X^\dag:=(I-[X(I-L^\dag L)]^\dag X)L^\dag \quad\mbox{and}\quad
\bbX:=X\,L_X^\dag;
\end{equation*}
see \cite{Elden:82,HansenBook:98}.  In terms of the GSVD components \eqref{eq:GSVD}, the
transformed $X$ is $\bbX=U\Gamma V'$. In particular, the diagonal elements of
$\Gamma=SM^\dag$ are the ordinary singular values of $\bbX$, but in reversed order.
Moreover, a PEER estimate can be obtained from a ridge-like penalization process with
respect to $\bbX$.  That is, for
\begin{equation}\label{eq:bbbeta}
\tilde{\mathbb{\bbbeta}}_{\alpha}
= \argmin_{\mathbb{\bbbeta}} \{||\mathbbm{y}-\bbX\mathbb{\bbbeta}||^2+\alpha|| \mathbb{\bbbeta}||^2 \}
\quad\mbox{where}\quad  \mathbbm{y} =[X(I-L^\dag L)]^\dag y,
\end{equation}
then
\begin{equation*}\label{eq:betaL_from_StdForm}
\tilde{\beta}_{\alpha,L} = L_X^\dag\tilde{\mathbb{\bbbeta}}_{\alpha}
+ \mathbbm{y}.
\end{equation*}
Note that the transformed estimate as given in terms of the GSVD factors is:
$\tilde{\bbbeta}_\alpha=VF\Gamma^\dag U'y$, where
$F=\diag\{\gamma_k^2/(\gamma_k^2+\alpha)\}$.

In what follows we consider invertible $L$ in which case
$L_X^\dag=L^{-1}$ and $[X(I-L^\dag L)]^\dag=0$.  In particular,
$\tilde{\beta}_{\alpha,L} = L^{-1}\tilde{\mathbb{\bbbeta}}_{\alpha}$.
For the penalty \eqref{eq:LQPQ} of the form $L=a(I-P_\mathcal{Q})+
bP_\mathcal{Q}$, then
$L^{-1}=\frac{1}{a}(I-P_\mathcal{Q})+\frac{1}{b}P_\mathcal{Q}$, and so
$\bbX=\frac{1}{a}X(I-P_\mathcal{Q})+\frac{1}{b}XP_\mathcal{Q}$.  The
regularization parameter, previously denoted by $\alpha$, can be
absorbed into the values $a$ and $b$, so we will denote this PEER
estimate $\tilde\beta_{\alpha,L}$ simply as $\tilde{\beta}_{a,b}$.

\begin{remark}\label{rmk:betaL_betaR}
When $a=b=\sqrt{\alpha}$, this is simply a ridge estimate:
$\tilde{\beta}_{a,b}=\tilde{\beta}_{\alpha,I}$.  Therefore, the best
performance among this family of estimates is as least as good as the
performance of ridge, regardless of the choice of $\mathcal{Q}$.
\end{remark}

\subsection{SVD targeted penalties}\label{sec:SVDdecompPen}

Consider the special case in which $\mathcal{Q}$ is the span of the $d$ largest right
singular vectors of an $n\times p$ matrix $X$ of rank $n$.  Let $X=U\matr{cc}{0 & D}V'$
be an ordinary singular value decomposition where $D$ is a diagonal matrix of singular
values. For consistency with the GSVD notation, these will be ordered as $0\le
\sigma_1\le\cdots\le \sigma_n$. As before, the first $p-n$ columns of $V$ are not used.
Rather than introduce extra notation, we write $X=UDV'$, letting $V$ denote the $n\times
p$ whose columns correspond to the singular vectors in $D$.  So now, the {\em last} $d$
columns of $V$  correspond to the $d$ largest singular values of $X$ (i.e., $Q=V_{\oo}$).

We are interested interested in the penalty $L=a(I-P_\mathcal{Q})+b P_\mathcal{Q}$, where
$d=\dim(\Null(I-P_\mathcal{Q}))$. Similar to before, set $\nd=n-d$ and define
$\nd\times\nd$ and $d\times d$ submatrices, $D_{\nd}$ and $D_{\oo}$, of $D$ as
\begin{equation}\label{eq:D_Lamba_decomp}
D=\matr{cc}{D_{\nd}&0\\ 0&D_{\oo}}; \quad\mbox{also set}\quad \Lambda=\matr{cc}{aI_{\nd} & 0\\ 0 & bI_d}.
\end{equation}

Here, $P_\mathcal{Q}=V_{\oo}{V_{\oo}}'$ and $(I-P_\mathcal{Q})=V_{\nd}{V_{\nd}}'$ and so,
\begin{equation*}
\begin{aligned}
\bbX & = \frac{1}{a}UDV'(V_{\nd}{V_{\nd}}') +\frac{1}{b}UDV'(V_{\oo}{V_{\oo}}')
       = \frac{1}{a}UD\matr{c}{{V_{\nd}}'\\0}+\frac{1}{b}UD\matr{c}{0\\ {V_{\oo}}'}\\
     & = U\matr{c}{\frac{1}{a}D_{\nd}{V_{\nd}}'\\0} +U\matr{c}{0 \\ \frac{1}{b}D_{\oo}{V_{\oo}}'}
       = U(D\Lambda^{-1}) V'
\end{aligned}
\end{equation*}
This decomposition implies that the ridge estimate in \eqref{eq:bbbeta} is of the
following form:  setting $G=D\Lambda^{-1}$, denoting its diagonal entries by
$\{\gamma_k\}$, and defining $F=\diag\{\gamma_k^2/(\gamma_k^2+1)\}$ gives
$\tilde\bbbeta=VFG^\dag U'y$. Now,
\begin{equation*}
L^{-1}V = \frac{1}{a}V_{\nd}{V_{\nd}}'V + \frac{1}{b}V_{\oo}{V_{\oo}}'V
        = \matr{c}{V_{\nd}\\ V_{\oo}}\matr{cc}{\frac{1}{a}I_{\nd} & 0\\ 0 & \frac{1}{b}I_d} = V\Lambda^{-1}
\end{equation*}
and so $\tilde{\beta}_{a,b}=L^{-1}\tilde\bbbeta = L^{-1}(V FG^\dag U'y) = V\Lambda^{-1} F
\Lambda D^{-1} U'y$.  By the decomposition
 \eqref{eq:D_Lamba_decomp},
 $$\tilde{\beta}_{a,b} = V_{\nd} F_{\nd} D_{\nd}^{-1} {U_{\nd}}'y + V_{\oo} F_{\oo} D_{\oo}^{-1}
 {U_{\oo}}'y.$$
This shows that the estimate decomposes as follows.
\begin{theorem}\label{thm:beta_aLQbPQ}
Let $\mathcal{Q}$ be the span of the largest $d$ right singular vectors of $X$.  Set
$L=a(I-P_\mathcal{Q})+ bP_\mathcal{Q}$.  Then, in terms of the notation above, the
estimate $\tilde{\beta}_{a,b}$ decomposes as
\begin{equation}\label{eq:beta_aLQbPQ}
 \tilde{\beta}_{a,b}
 = \sum_{k=1}^{n-d} \left(\frac{\sigma_k^2}{\sigma_k^2+a^2}\right) \frac{1}\sigma_k u_k'y \, v_k
 + \sum_{k=d+1}^n  \left(\frac{\sigma_k^2}{\sigma_k^2+b^2}\right) \frac{1}\sigma_k u_k'y \, v_k,
\end{equation}
where the left and right terms are independent of $b$ and $a$, respectively.
\end{theorem}

Similar arguments can be used to decompose an estimate for arbitrary
$\mathcal{Q}$: 
\begin{equation}\label{eq:beta_aLQbPQ_W}
 \tilde{\beta}_{a,b}
 = \sum_{k=1}^{n-d} \left(\frac{\sigma_k^2}{\sigma_k^2+a^2\mu_k^2}\right) \frac{1}\sigma_k u_k'y \, w_k
 + \sum_{k=d+1}^n  \left(\frac{\sigma_k^2}{\sigma_k^2+b^2\mu_k^2}\right) \frac{1}\sigma_k u_k'y \, w_k.
\end{equation}
In this case, however, all terms are dependent on both $a$ and $b$. Indeed,
using notation as in \eqref{eq:beta_aLQ} one can decompose
$\bbX=\frac{1}{a}U\Gamma_{\nd} V' + \frac{1}{b}U\Gamma_{\oo} V'$ and obtain
$\tilde\bbbeta=VF\Gamma^\dag U'y$.  However,  $L^{-1}V=WM^\dag$, and the
non-orthogonal terms provided by $W$ do not decompose the estimate into terms
from the orthogonal sum $\mathcal{Q} \oplus \mathcal{Q}^\perp$.

The following corollary, along with Remark~\ref{rmk:betaL_betaR},
records the manner in which \eqref{eq:beta_aLQbPQ} is a family of
penalized estimates, parameterized by $a,b>0$ and $d\in \{1,...,n\}$,
that extends some standard estimates.

\begin{corollary}\label{cor:beta_ab_family}
Under the conditions in Theorem~\ref{thm:beta_aLQbPQ},
\begin{enumerate}\itemsep=-3pt
    \item when $a>b>0$, \ $\tilde{\beta}_{a,b}$ is a sum of
        weighted ridge estimates on $\mathcal{Q}$ and
        $\mathcal{Q}^\perp$;
    \item when $a>0$ and $b=0$, $\tilde{\beta}_{a,0}$ is given by
        \eqref{eq:beta_aLQ}, which is a sum of PCR and ridge estimates on
        $\mathcal{Q}$ and $\mathcal{Q}^\perp$, respectively;
    \item for each $d$, the PCR estimate $\tilde{\beta}_{\mbox{\tiny PCR}}^d$ is the
        limit of $\tilde{\beta}_{a,0}$ as $a\to\infty$.
\end{enumerate}
\end{corollary}
In item 2, this estimate is similar to PCR except that a ridge penalty is placed on the
least-dominant singular vectors. Under the assumptions here, $w_k\equiv v_k$ are the
ordinary singular vectors of $X$ and the ordinary singular values appear as
$\gamma_k=\sigma_k/\mu_k$, for $\mu_k>0$.  In the second term of \eqref{eq:beta_aLQ}, the
singular vectors are in the null space of $L$ (since $b=0$), and so $\mu_k=0$ and
$\sigma_k=1$, for $k=n-d+1,...,p$. Regarding item 3, although a PCR estimate is not
obtained from equation~\eqref{eq:beta_NormalEqns} for any $L$, it is a limit of such
estimates.

Other decompositions may be obtained simply by using a permutation, such as $Q=\Pi V$,
for some $n\times n$ permutation matrix $\Pi$. Stein's estimate,
$\tilde\beta_{\alpha,\mbox{\tiny S}}$, also fits into this framework as follows.

When $X'X$ is nonsingular, then $\tilde\beta_{\alpha,\mbox{\tiny S}}= (X'X+\alpha
X'X)^{-1}X'y$ (see, e.g., the class `STEIN' in \cite{DemSchWer:77}), and $X'X=VD'DV'$.
Hence this estimate arises from the penalty $L_{\mbox{\tiny S}}=DV'$. This is a
re-weighted version of $L=a(I-P_\mathcal{Q})$ where $d=n$, $Q=V$ and the parameter $a$ is
replaced by the matrix $D$. The result is a constant filter factor
$F=\diag\{1/(1+\alpha)\}$.  Using $d<n$ and $Q=V_d$ is a natural extension of this idea.
More generally, $\mathcal{Q}$ may be enriched with functions that span a wider range of
structure potentially relevant to the estimate.
This concept is illustrated in Section~\ref{sec:RamanApplic} where instead of $V_d$, we
use a $d$-dimensional set of experimentally-derived ``template" spectra supplemented with
their derivatives to define $\mathcal{Q}$.

As an aside, we note that in a different approach to regularization one can
define a general family of estimates arising from the SVD by way of
$\tilde{\beta}_{h,\varphi}=V\Sigma_hU'y$, where
$\Sigma_h=\diag\{\frac{\sigma_k}{h}\varphi(\frac{\sigma_k^2}{h^2}) \}$, and
$\varphi:\bbR_+\to\bbR$ is an arbitrary continuous function \cite{Neum:98}.  A
ridge estimate is obtained for $\varphi(t)=1/(1+t)$, and PCR obtained for
$\varphi(t)=1/t$ if $t>1$, $\varphi(t)=0$ if $t\le 1$ (an $L^2$-limit of
continuous functions). This is similar to item~3 in
Corollary~\ref{cor:beta_ab_family}, but the family of estimates
$\tilde{\beta}_{h,\varphi}$ is formulated in terms of functional filter factors
rather than explicit penalty operators. Related to this is the fact that the
optimal (with respect to MSE) estimate  using SVD filter factors is, in the
case $C=\sigma_\epsilon I$, expressed as
 $\tilde{\beta}_{\mbox{\tiny OH}}=VFD^\dag V'y$, where
$F=\diag\{\sigma_k^2/(\sigma_k^2+\sigma_\epsilon^2(v_k'\beta)^{-2}) \}$; see the ``ideal
filter" of \cite{ObrienHolt:81}.  In fact, it's easy to check that this optimal estimate
can be obtained as $\tilde{\beta}_{\mbox{\tiny OH}}=\tilde{\beta}_{\alpha,L}$ for some
$L\neq I$. Since $\tilde{\beta}_{\mbox{\tiny OH}}$ involves knowledge of $\beta$, it is
not directly obtainable but it points to the optimality of a PEER estimate.

\subsection{The MSE of some penalized estimates.}\label{sec:MSEcompare}

Theorem~\ref{thm:beta_aLQbPQ} is used here to show that $\tilde{\beta}_{a,b}$
can have smaller MSE than the ridge or PCR estimates for a wide range of values
of $a$ and/or $b$.  The MSE is potentially decreased further when $L$ is
defined by a more general $\mathcal{Q}$.   In that case, a general statement is
difficult to formulate but Proposition~\ref{prop:MSE_perturb} confirms that any
improvement in MSE is robust to perturbations in $L$ (e.g., general
$\mathcal{Q}$) and errors in $x$.

An immediate consequence of Theorem~\ref{thm:beta_aLQbPQ} is that the
mean squared error for an estimate in this family
\eqref{eq:beta_aLQbPQ} decomposes into easily-identifiable terms for
the bias and variance:
\begin{equation}\label{eq:MSE_beta_abV}
\begin{aligned}
 \mbox{MSE}(\tilde{\beta}_{a,b}) = \sigma_\epsilon^2
   & \sum_{k=1}^{n-d} \frac{\sigma_k^2}{(\sigma_k^2+ a^2)^2}
     +  \sum_{k=1}^{n-d}  \left(\frac{a^2}{\sigma_k^2+ a^2}\right)^2  (v_k'\beta)^2   \\
   & + \sigma_\epsilon^2 \sum_{k=n-d+1}^{n} \frac{\sigma_k^2}{(\sigma_k^2+ b^2)^2}
     + \sum_{k=n-d+1}^{n} \left(\frac{b^2}{\sigma_k^2+ b^2}\right)^2  (v_k'\beta)^2.
\end{aligned}
\end{equation}
The influence of $b=0$ on the estimate is now clear: when the numerical rank of $X$ is
small relative to $d$, the $\sigma_k$'s in the third term decrease and the contribution
to the variance from this term increases---the estimate fails for the same reason that
ordinary least-squares fails.  Any nonzero $b$ stabilizes the estimate in the same way
that a nonzero $\alpha$ stabilizes a standard ridge estimate; the decomposition
\eqref{eq:beta_aLQbPQ} merely re-focuses the penalty. This is illustrated in
Section~\ref{sec:NumProps} (Table~\ref{tab:estL})  and  in the Appendix
(Table~\ref{tab:estF}). Although there are three parameters to consider, the MSE of
$\tilde{\beta}_{a,b}$ is relatively insensitive to $b>0$ for sufficiently large $d$. This
could be optimized (similar to efforts to optimize the number of principal components)
but here we assume some knowledge regarding $\mathcal{Q}$, hence $d$. Relationships
between ridge, PCR and PEER estimates in this family $\{\tilde{\beta}_{a,b}\}_{a,b>0}$
can be quantified more specifically as follows.

\begin{proposition}\label{prop:MSE_Ridge_beta}
Suppose $\beta\in\mathcal{Q}$ and fix $\alpha>0$. Then for any $a>\sqrt{\alpha}$, the
ridge estimate satisfies
\begin{equation*} \mbox{MSE}(\tilde{\beta}_{\alpha,I})
  \equiv \mbox{MSE}(\tilde{\beta}_{\sqrt{\alpha},\sqrt{\alpha}})
  > \mbox{MSE}(\tilde{\beta}_{a,\sqrt{\alpha}}).
\end{equation*}
\end{proposition}
\begin{proof} This follows from the fact that  ${V_{\nd}}'\beta=0$ and so the second
term in \eqref{eq:MSE_beta_abV} is zero.  Therefore, the contribution to the
MSE by the first term is decreased whenever $a>\sqrt{\alpha}$.
\end{proof}

If $\beta$ is exactly a sum of the $d$ dominant right singular vectors, A PCR estimate
using $d$ terms may perform well, but is not optimal:
\begin{proposition}\label{prop:MSE_PCR_beta}
If $\beta\in\mathcal{Q}$, a sufficient condition for the  PCR estimate to satisfy
\begin{equation*}\label{eq:suff_MSE_PCR I}
\mbox{MSE}({\tilde{\beta}^d}_{\mbox{\tiny PCR}})
  \equiv \mbox{MSE}(\tilde{\beta}_{\infty,0})
   > \mbox{MSE}(\tilde{\beta}_{\infty,b})
\end{equation*}
is
\begin{equation}\label{eq:suff_MSE_PCR_II}
  \sigma_\epsilon^2\left( \sum_{k=n-d+1}^{n} \frac{1}{\sigma_k^2}+ \frac{2d}{b^2}\right)
  >  ||{V_{\oo}}'\beta||^2
\end{equation}
\end{proposition}
Note that the left side of \eqref{eq:suff_MSE_PCR_II} increases without bound as
$\sigma_k\to 0$. Since $||{V_{\oo}}'\beta||^2=\sum_{k=n-d+1}^{n} (v_k'\beta)^2$, and
since the premise of PCR is that $v_k'\beta$ decreases with decreasing $\sigma_k$, this
sufficient condition is entirely plausible.

\begin{proof}  The MSE of ${\tilde{\beta}^d}_{\mbox{\tiny PCR}}$ consists of the second and third terms of
\eqref{eq:MSE_beta_abV}:
\begin{equation*}
 \mbox{MSE}({\tilde{\beta}^d}_{\mbox{\tiny PCR}})=
      \sum_{k=1}^{n-d} (v_k'\beta)^2
      + \sigma_\epsilon^2 \sum_{k=n-d+1}^{n} \frac{1}{\sigma_k^2}.
\end{equation*}
 In particular,
a sufficient condition for this to exceed
$\mbox{MSE}(\tilde{\beta}_{\infty,b})$ is for the variance term to exceed the
third and fourth terms of \eqref{eq:MSE_beta_abV}:
\begin{equation*}
\sigma_\epsilon^2 \sum_{k=n-d+1}^{n} \frac{1}{\sigma_k^2} > \sigma_\epsilon^2 \sum_{k=n-d+1}^{n} \frac{\sigma_k^2}{(\sigma_k^2+ b^2)^2}
     + \sum_{k=n-d+1}^{n} \left(\frac{b^2}{\sigma_k^2+ b^2}\right)^2  (v_k'\beta)^2.
\end{equation*}
Its easy to check that this is satisfied when
\eqref{eq:suff_MSE_PCR_II} holds.
\end{proof}

A comment by Bingham and Larntz \cite{BingLarntz:77} on the intensive simulation study of
ridge regression in \cite{DemSchWer:77} notes that ``it is not at all clear that ridge
methods offer a clear-cut improvement over [ordinary] least squares except for particular
orientations of $\beta$ relative to the eigenvectors of $X'X$."
Equation \eqref{eq:MSE_beta_abV} repeats this observation relating these two classical
methods as well as the minor extensions contained in \eqref{eq:beta_aLQbPQ}. If, on the
other hand, ``the orientation of $\beta$ relative to the [$v_k$'s]" is not favorable,
i.e., if $\beta$ is nowhere near the range of $V$, then a PEER estimate as in
\eqref{eq:beta_aLQbPQ_W} is more desirable than \eqref{eq:beta_aLQbPQ} (assuming
sufficient prior knowledge).

In summary, the family of estimates $\{\tilde{\beta}_{a,b}\}_{a,b>0}$ in
\eqref{eq:beta_aLQbPQ} represents a hybrid of ridge and PCR estimation.  This
family---based on the ordinary singular vectors of $X$---is introduced here to provide a
framework within which these two familiar estimates can be compared to (slightly) more
general PEER estimates. Direct analytical comparison between general PEER estimates is
more difficult since there's no simple relationship between the generalized singular
vectors for two different $L$ (including $L=I$ versus $L\neq I$). However, it is
important that the estimation process be stable with respect to changes in $L$ and/or
$X$.  I.e., in going from an estimate in \eqref{eq:beta_aLQbPQ} to one in
\eqref{eq:beta_aLQbPQ_W}, the performance of the estimate should be predictably altered.
Given an estimate in Proposition~\ref{prop:MSE_Ridge_beta}, if $\mathcal{Q}$ is modified
and/or $X$ is observed with error, the MSE of the corresponding estimate,
$\tilde{\beta}^E_{\alpha,L}$, should be controlled: for sufficiently small perturbation
$E$, the corresponding estimate $\mbox{MSE}(\tilde{\beta}^E_{\alpha,L})$ should be close
to $\mbox{MSE}(\tilde{\beta}_{\alpha,I})$. This ``stability" is true in general. To see
this, recall $Z=\matr{cc}{X' &\sqrt{\alpha} L'}'$ (of rank $p$) and
$\underline{y}=\matr{cc}{y'& 0}'$. Then another way to represent the estimate
\eqref{eq:beta_NormalEqns} is $\tilde{\beta}_{\alpha,L}=Z^\dag \underline{y}$. Let
$E=\matr{cc}{{E_1}' & {E_2}'}'$ for some $n\times p$ and $m\times p$ matrices $E_1$ and
$E_2$. Set $Z_E=Z+E$ and denote the perturbed estimate by
$\tilde{\beta}^E_{\alpha,L}=Z_E^\dag \underline{y}$.  By continuity of the generalized
inverse (e.g., \cite{Bjorck:96}, Section~1.4), $\lim_{||E||\to 0} Z_E^\dag=Z^\dag$ if and
only if $\lim_{||E||\to 0}\rank(Z_E)=\rank(Z)$. Therefore, provided the rank of $Z$ is
not changed by $E$,
\begin{equation*}
\lim_{||E||\to 0}||\tilde{\beta}_{\alpha,L}-{\tilde{\beta}^E}_{\alpha,L}||\le
\lim_{||E||\to 0} ||Z^\dag-Z_E^\dag||\, ||y||= 0,
\end{equation*}
and hence $\mbox{MSE}({\tilde{\beta}^E}_{\alpha,L}) \to
\mbox{MSE}({\tilde{\beta}}_{\alpha,L})$ as $||E||\to 0$.  A more
specific bound on the difference of estimates can be obtained under the
condition $||Z^\dag|| ||E||<1$ which implies that $||Z_E^\dag||<
\frac{||Z^\dag||}{1-||Z^\dag|| ||E||}$.  This can be used to obtain the
following  bound.
\begin{proposition}\label{prop:MSE_perturb}
Assume $||Z^\dag|| ||E||<1$ and let
$r=\underline{y}-Z\tilde{\beta}_{\alpha,L}$.  Then
\begin{equation*}
||\tilde{\beta}_{\alpha,L}-{\tilde{\beta}^E}_{\alpha,L}||\le
\frac{||Z^\dag|| ||E|| }{1-||Z^\dag|| ||E||}\left(||\tilde{\beta}_{\alpha,L}||+||Z^\dag|| ||r|| \right).
\end{equation*}
\end{proposition}
See \cite{Bjorck:96} and \cite{Hansen:89}.

\section{Tuning parameter selection} \label{sec:TuningParam}

Despite our focus on the GSVD, the computation of a PEER estimate in \eqref{eq:PenEst}
does not, of course, require that this decomposition be computed. Rather, the role of the
GSVD has been to provide analytical insight into the role a penalty operator plays in the
estimation process. For computation, on the other hand, we have chosen to use a method in
which the tuning parameter, $\alpha$, is estimated as part of the coefficient-function
estimation process.

Because the choice of tuning parameter is so important, many selection criteria have been
proposed, including generalized cross-validation (GCV) \cite{CravenWahba:79}, AIC and its
finite sample corrections \cite{Wood:06}. As an alternative to GCV and AIC, a
recently-proven equivalence between the penalized least squares estimation and a linear
mixed model (LMM) representation \cite{BrumRuppWand:99} can be used. In particular, the
best linear unbiased predictor (BLUP) of the response $y$ is composed of the best linear
unbiased estimator of the fixed effects and BLUP of the random effects for the given
values of the random component variances (see \cite{Speed:91} and
\cite{BrumRuppWand:99}). Within the LMM framework, restricted maximum likelihood (REML)
can be used to estimate the variance components and thus the choice of the tuning
parameter, $\alpha$, which is equal to the ratio of the error variance and the random
effects variance \cite{RuppertWandCarroll:03}.
 REML-based estimation of the tuning
parameter has been shown to perform at least as well as the other criteria and, under
certain conditions, it seen to be less variable than GCV-based estimation
\cite{ReissOgden:09}. In our case, the penalized least-squares criterion
\eqref{eq:PenEst} is equivalent to
\begin{equation}
\tilde\beta_{\alpha,L}=\argmin_\beta\{||y-X_{unp} \beta_{unp} - X_{pen} \beta_{pen}||^2
+ \alpha||L\beta||^2\}
\label{eq:mix}
\end{equation}
where $\beta = [\beta_{unp}' \ \beta_{pen}']'$, the $X_{unp}$ corresponds to the
unpenalized part of the design matrix, and $X_{pen}$ to the penalized part.

For simplicity of presentation, we describe the transformation with an invertible $L$.
However, a generalized inverse can be used in case $L$ is not of full rank; see
equation~\eqref{eq:bbbeta}.  Also, to facilitate a straightforward use of existing linear
mixed model routines in widely available software packages (e.g., \proglang{R}
\cite{Rproject:11} or \proglang{SAS} software \cite{SAS:2008}), we transform the
coefficient vector $\beta$ using the inverse of the matrix $L$. Let $X^{\star} = X L$ and
$\beta^{\star} = L^{-1} \beta$. Then equation \eqref{eq:mix} can be modified as follows
\begin{equation*}
\tilde\beta_{\alpha,L}^{\star} = \argmin_\beta\{||y - X^{\star} \beta^{\star}||^2
+ \alpha||\beta^{\star}||^2\}.
\end{equation*}
This REML-based estimation of tuning parameters is used in the application of
Section~\ref{sec:RamanApplic}.

For estimation of the parameters $a$, $b$ and $\alpha$ involved in the
decomposition-based penalty of equation~\eqref{eq:beta_aLQbPQ}, we view $a$ and $b$ as
weights in a tradeoff between the subspaces and assume $ab=$ const.  For implementation,
we fix one, estimate the other using a grid search, and use REML to estimate $\alpha$.

\section{Numerical examples}\label{sec:NumProps}
To illustrate algebraic properties given in Section~\ref{sec:AnalProps}, we consider PEER
estimation alongside some familiar methods in several numerical examples.
Section~\ref{sec:BumpsSim} elaborates on the simple example in
Section~\ref{sec:BumpsIntro}.  These mass spectrometry-like predictors are mathematically
synthesized in a manner similar to the study of Reiss and Ogden \cite{ReissOgden:07} (see
also a numerical study in \cite{StoutKalivas:06}).  Here, $\beta$ is also synthesized to
represent a spectrum, or specific set of bumps.  In contrast,
Section~\ref{sec:RamanApplic} presents a real application to Raman spectroscopy data in
which a set of spectra $\{x_i\}$ and nanoparticle concentrations $\{y_i\}$ are obtained
from sets of laboratory mixtures.  This laboratory-based application is preceded in
section~\ref{sec:RamanSim} by a simulation that uses these same Raman spectra. In both
Raman examples, targeted penalties \eqref{eq:LQPQ} are defined using discretized
functions $q_j$ chosen to span specific subspaces, $\mathcal{Q}=\span\{q_j\}_{j=1}^d$. As
before, let $Q=\col[q_1,...,q_d]$ and $P_\mathcal{Q}=QQ^\dag$.

Each section displays the results from several methods, including derivative-based
penalties. Implementing these requires a choice of discretization scheme and boundary
conditions which define the operator. We use $\mathcal{D}^2$ where
$\mathcal{D}=[d_{i,j}]$ is a square matrix with entries $d_{i,i}=1$, $d_{i,i+1}=-1$ and
$d_{i,j}=0$ otherwise.  In addition to some standard estimates,
sections~\ref{sec:RamanApplic} and \ref{sec:RamanSim} also consider $\mbox{FPCR}_R$, a
functional PCR estimate described in \cite{ReissOgden:07}.  This approach extends the
penalized B-spline estimates of \cite{CarFerSar:03} and assumes $\beta = B\eta$ where $B$
is an $p\times K$ matrix whose columns consist of $K$ B-spline functions and $\eta$ is a
vector of B-spline coefficients. The estimation process takes place in the coefficient
space using the penalty $L=\mathcal{D}^2$ applied to $\eta$. The $\mbox{FPCR}_R$ estimate
further assumes $\beta = BV_d\,\eta$ ($V_d$ as defined in section~\ref{sec:BkgdNotn}).

Estimation error is defined as mean squared error (MSE)
$||\beta-\tilde{\beta}_{\alpha,L}||^2$, and the prediction error defined similarly as
$\sum_i |y_i - \tilde{y_i}|^2$, where $\tilde{y_i}=\langle x_i,\tilde\beta\rangle$. Each
simulation incorporates response random errors, $\epsilon_i \sim
N(0,\sigma_{\epsilon}^2)$, added to the $i$th true response, $y^{\mbox{\tiny
true}}_i=\langle x_i,\beta\rangle$. Letting $S_Y^2$ denote the sample variance in the set
$\{y^{\mbox{\tiny true}}_i\}_{i=1}^n$, the response random errors, $\epsilon_i$, are
chosen such that $R^2 := S_Y^2/[S_Y^2 + \sigma^2_{\epsilon}]$ (the squared multiple
correlation coefficient of the true model) takes values 0.6 and 0.8. In
sections~\ref{sec:BumpsSim} and~\ref{sec:RamanSim}, tuning parameters are chosen by a
grid search. In section~\ref{sec:RamanApplic}, tuning parameters are chosen using REML,
as described in section~\ref{sec:TuningParam}.

\subsection{Bumps Simulation} \label{sec:BumpsSim}
Here we elaborate on the simple example of section~\ref{sec:BumpsIntro}. This simulation
involves bumpy predictor curves $x_i(t)$ with a response $y_i$ that depends on the
amplitudes $x_i(t)$ at some of the bump locations, $t=c_k$, via the regression function
$\beta$. In particular,
\begin{eqnarray*}
x_i(t) = \sum_{j \in J_X} a_{ij} \exp[-b_{j}(t-c_j)] + e_i(t), \quad
\beta(t) = \sum_{j \in J_{\beta}}  a_j \exp[-b_j(t-c_j)], \ \ t \in
[0,1]
\end{eqnarray*}
where $J_X = \{2,6,10,14,20,26,30\}$, $J_{\beta} = \{6,14,26\}$, $a_{\star}$ are the
magnitudes, $b_{\star}$ are the spreads, and $c_{\star}$ are the locations of the bumps.
In the first simulation, we set $b_j = 10000$ and $c_j=0.004(8j-1)$, the same for each
curve $x_i$.  This mimics, for instance, curves seen in mass spectrometry data. The
assumption $J_{\beta}\subset J_X$ simulates a setting in which the response is associated
with a subset of metabolite or protein features in a collection of spectra. The
$a_{ij}$'s are from a uniform distribution, and $a_j = 3, 5, 2$ for $j=6,14,24$,
respectively. We consider discretized curves, $x_i(t)$, evaluated at $p=250$ points,
$t_j$, $j=1, \ldots, p$. The sample size is fixed at $n=50$ in each case.

\noindent{\bf Penalties.}  We consider a variety of estimation procedures: ridge ($L=I$),
second-derivative ($\mathcal{D}^2$), a more general derivative operator
($\mathcal{D}^2+a\,I$) and PCR.  We also define two decomposition-based penalties
\eqref{eq:LQPQ} formed by specific subspaces $\mathcal{Q}=\span\{q_j\}_{j\in J}$ for
$q_j$ of the form $q_j(t)=a_j \exp[b_j(t-c_j)]$, with $c_j$ at all locations seen in the
predictors, $J_V=\{2,6,10,14,20,26,30\}$, or at uniformly-spaced locations,
$J_U=\{2,4,\ldots,30\}$; denote these penalties by $L_V$ and $L_U$, respectively.

\noindent{\bf Simulation results.}  The simulation incorporates two sources of noise: (i)
response random errors, $\epsilon_i \sim N(0,\sigma_{\epsilon}^2)$, added to the $i$th
true response so that $R^2=0.6, 0.8$; (ii) measurement error, $e_i\sim N_p(0,\sigma_{e}^2
I)$, added to the $i$th predictor, $x_i$.  To define a signal-to-noise ratio, $S/N$,  set
$S_i^2:=||x_i-\mu_i||^2/(p-1)$, where $\mu_i$ is the mean value of $x_i$, and set
$S_X^2:=1/n\sum_i S_i^2$.  The $e_i$ are chosen so that $S/N:=S_X/\sigma_e = 2,5,10$.

Figure~\ref{Fig:PartSums} shows a few partial sums of \eqref{eq:betaGSVE} for estimates
arising from three penalties:  $\mathcal{D}^2$, $L=I$ and $L_V$, when $R^2= 0.8$ and
$S/N=2$. Table~\ref{tab:estL} gives a summary of estimation errors. The penalty $L_V$,
exploiting known structure, performs well in terms of estimation error.  Not
surprisingly, a penalty that encourages low-frequency singular vectors, $\mathcal{D}^2$,
is a poor choice although $\mathcal{D}^2+a\,I$ easily improves on $\mathcal{D}^2$ since
the GSVs are more compatible with the relevant structure.   PCR performs well with
estimation errors that can be several times smaller than those of ridge.   The number of
terms used in PCR ranges here from 8 ($S/N=10$) to 25 ($S/N=2$).

\begin{table}[h!bt]
\begin{center}
\caption[]{\small Estimation errors (MSE) for simulation with selected bump
locations. Sample size is $n=50$.} \label{tab:estL}
  \begin{tabular}{lr||rr|rrrrr} \hline
\multicolumn{1}{l}{$R^2$} &
\multicolumn{1}{l||}{$S/N$} &
\multicolumn{1}{r}{$L_V$} &
\multicolumn{1}{r}{$L_{U}$} &
\multicolumn{1}{r}{\small PCR} &
\multicolumn{1}{r}{\small ridge} &
\multicolumn{1}{c}{$\mathcal{D}^2$} &
\multicolumn{1}{r}{$\mathcal{D}^2+a\,I$}  &
\\ \hline %
 0.8 & 10  &   4.00   &  13.81   &  9.38  &  34.39  &  359.83  &  76.31   \\
 0.8 &  5  &   3.72   &  15.46   &  21.50 &  40.02  &  246.17  &  72.64   \\
 0.8 &  2  &   4.40   &  12.96   &  57.89 &  58.22  &  126.75  &  59.35   \\
 0.6 & 10  &   9.60   &  21.60   &  14.10 &  50.50  &  497.70  &  113.60  \\
 0.6 &  5  &  10.22   &  21.65   &  26.02 &  50.68  &  338.70  &  87.58   \\
 0.6 &  2  &  11.75   &  23.18   &  63.50 &  67.94  &  181.75  &  78.45   \\
 \hline
\end{tabular}
\end{center}
\end{table}

Predictably, PCR performance degrades with decreasing $S/N$, a property that is
less pronounced, or not shared, by other estimates.  Performances of $L_V$ and
$L_U$ illustrate properties described in Section~\ref{sec:MSEcompare}.  As
$S/N\to 0$, the ordinary singular vectors of $X$ (on which ridge and PCR rely)
decreasingly represent the structure in $\beta$.  The GS vectors of $(X,L_V)$
and $(X,L_U)$, however, retain structure relevant for representing $\beta$.

\begin{table}[hbt]
\begin{center}
\caption[]{\small Prediction errors for simulation with selected bump
locations. Sample size is $n=50$. Errors are multiplied by 1000 for display.}
\label{tab:predL}
  \begin{tabular}{lr||rr|rrrc} \hline
\multicolumn{1}{l}{$R^2$}&
\multicolumn{1}{l||}{$S/N$}&
\multicolumn{1}{r}{$L_V$}&
\multicolumn{1}{r}{$L_U$}&
\multicolumn{1}{r}{\small PCR} &
\multicolumn{1}{r}{\small ridge} &
\multicolumn{1}{r}{$\mathcal{D}^2$}&
\multicolumn{1}{r}{$\mathcal{D}^2+a\,I$} \\ \hline
 0.8  &  10  &  9.0   &  10.5  &  10.8 & 16.6  &  19.3  &  12.9    \\
 0.8  &  5   &  8.4   &  11.0  &  12.2 & 26.7  &  27.9  &  17.8   \\
 0.8  &  2   &  12.9  &  19.0  &  53.2 & 55.7  &  50.3  &  40.1   \\
 0.6  &  10  &  21.4  &  23.0  &  23.9 & 33.0  &  39.0  &  26.2    \\
 0.6  &  5   &  23.9  &  25.0  &  29.5 & 49.2  &  54.6  &  34.4    \\
 0.6  &  2   &  34.4  &  42.5  &  90.4 & 110.4 &  104.4 &  77.9   \\
\hline
\end{tabular}
\end{center}
\end{table}

Table~\ref{tab:predL} summarizes prediction errors.  When $S/N$ is large,
performance of PCR is comparable with $L_V$ and $L_U$, but degrades for low
$S/N$. Here, even $\mathcal{D}^2+a\,I$ provides smaller prediction errors, in
most cases, than ridge, $\mathcal{D}^2$ or PCR. This illustrates the GS vectors
role in \eqref{eq:GoutisPenEq} and reiterates observations in \cite{Goutis:98}.

\subsection{Raman simulation} \label{sec:RamanSim}
We consider Raman spectroscopy curves which represent a vibrational response of
laser-excited co-organic/inorganic nanoparticles (COINs). Each COIN has a unique
signature spectrum and serves as a sensitive nanotag for immunoassays; see
\cite{Lutz:08,Schachaf:09}.  Each spectrum consists of absorbance values measured at
$p=600$ wavenumbers.  By the Beer-Lambert law, light absorbance increases linearly with a
COIN's concentration and so a spectrum from a mixture of COINs is reasonably modeled by a
linear combination of pure COIN spectra. The data here come from experiments that were
designed to establish the ability of these COINs to measure the existence and abundance
of antigens in single-cell assays.

Let $P_{1},...,P_{10}$ denote spectra from nine pure COINs and one ``blank" (no
biochemical material), each normalized to norm one.
We form in-silico mixtures as follows:
 $ x_i = \sum_{k=1}^{10} c_{i,k}P_{k}$, $i=1,...,n$, with
coefficients $\{c_{i,k}\}$ generated from a uniform distribution. Figure~\ref{Fig:COINs}
shows representative spectra from all nine COINs superimposed on a collection of mixture
spectra, $\{x_i\}_{i=1}^{50}$.  Included in  Figure~\ref{Fig:COINs} is the $\beta$
(dashed curve) used to defined the simulation:  $y_i=\langle x_i,\beta\rangle +
\epsilon$, $\epsilon\thicksim N(0,\sigma^2)$.

\begin{figure}
\centering
{\includegraphics[height=6cm,width=10cm]{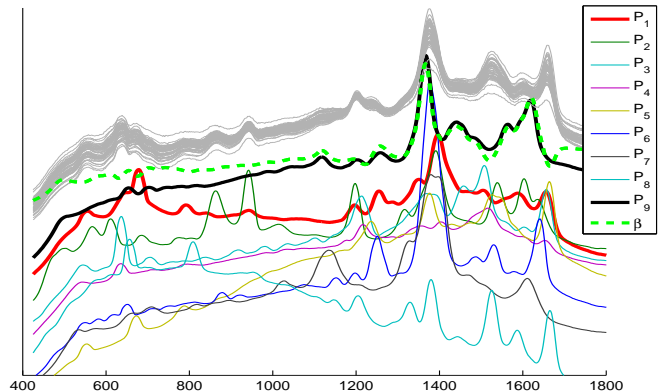}}
\vskip-10pt\caption{\small  Nine pure COIN spectra, $P_1,...,P_9$,
and a coefficient function, $\beta$ (each shifted for display).
$\beta$ arises as a solution to the fLM in which $y$ denotes concentrations of $P_9$
in an in silico mixture of $50$ COIN spectra, $x_i$ (light gray).
This $\beta$ is used in the simulation study of Section~\ref{sec:RamanSim}.}
\label{Fig:COINs}
\end{figure}

 In this simulation, we have created a coefficient function which, instead
of being modeled mathematically, is a curve that exhibits structure of the type found in
Raman spectra. Details on the construction of this $\beta$ are in
Appendix~\ref{App:RamanSimBeta} so here we simply note that it arises as a ridge estimate
from a set of in-silico mixtures of Raman spectra in which one COIN, $P_9$, is varied
prominently relative to the others. See Figure~\ref{Fig:COINs}. Motivation for defining
$\beta$ in this way is based on a view that it seems implausible for us to predict the
structure of realistic signal in these data and recreate it using polynomials, Gaussians
or other analytic functions.

Regardless of its construction, $\beta$ defines signal that allows us to compute
estimation and prediction error. The performances of five methods are summarized in
Table~\ref{tab:RamanSimulation}.  Note that although $\beta$ was constructed as a ridge
estimate (using a different set of in-silico mixtures; see
Appendix~\ref{App:RamanSimBeta}), the ridge penalty is not necessarily optimal for
recovering $\beta$. This is because the strictly empirical eigenvectors associated with
the new spectra may contain structure not informative regarding $y$. Also, in these data,
the performance of FPCR$_R$ is adversely affected by a tendency for the estimate to be
smooth; cf., Figure~\ref{Fig:AOH_BetaHats}. The PEER penalty used here is defined by a
decomposition-based operator \eqref{eq:beta_aLQbPQ} in which $\mathcal{Q}$ is spanned by
a 10-dimensional set of pure-COIN spectra (including a blank). The success of such an
estimate obviously depends on an informed formation of $\mathcal{Q}$, but as long as the
parameter-selection procedure allows for $a=b$, then the set of possible estimates
includes ridge as well as estimates with potentially lower MSE than ridge; see
Proposition~\ref{prop:MSE_Ridge_beta}.

\begin{table}[hbt]
\begin{center}
\caption[]{\small Estimation (MSE) and prediction (PE) errors of several penalization
methods for the simulation described in Figure~\ref{Fig:COINs}. Numbers represent the
average error from 100 runs. PE errors are multiplied by 1000 for display.}
\label{tab:RamanSimulation}
  \begin{tabular}{l||r|rrrr} \toprule
\multicolumn{1}{r}{} &
\multicolumn{1}{r}{$L_Q$} &
\multicolumn{1}{r}{\small PCR} &
\multicolumn{1}{r}{\small ridge} &
\multicolumn{1}{c}{$\mathcal{D}^2+a\,I$} &
\multicolumn{1}{c}{FPCR$_R$} \\
\toprule
 MSE &  8.91  &   12.34  & 13.87  &   41.69  & 15.29 \\
 PE  & 0.0071 &  0.0179  &  0.0139  &    0.0131  & 0.0175\\
\bottomrule
\end{tabular}
\end{center}
\end{table}

We note that this simulation may be viewed as inherently unfair since the PEER estimate
uses knowledge about the relevant structure.  However, this is a point worth
reemphasizing: when prior knowledge about the structure of the data is available, it can
be incorporated naturally into the regression problem.

\subsection{Raman Application} \label{sec:RamanApplic}

We now consider spectra representing true antibody-conjugated COINs from nine laboratory
mixtures.  These mixtures contain various concentrations of eight COINs (of the nine
shown in Figure~\ref{Fig:COINs}).  Spectra from four technical replicates in each mixture
are included to create a set of $n=36$ spectra $\{x_i\}_{i=1}^n$. We designate $P_1$ as
the COIN whose concentration within each mixture defines $y$.   Assuming a linear
relationship between the spectra, $\{x_i\}$, and the $P_1$-concentrations, $\{y_i\}$, we
estimate $P_1$. More precisely, we estimate the structure in $P_1$ that correlates most
with its concentrations, as manifest in this set of mixtures. The fLM is a simplistic
model of this relationship between the concentration of $P_1$ and its functional
structure, but the physics of this technology imply it is a reasonable starting point.

 We present the results of three estimation methods: ridge, FPCR$_R$ and
PEER.  In constructing a PEER penalty, we note that the informative structure in Raman
spectra is not that of low-frequency or other easily modeled features, but it may be
obtainable experimentally.  Therefore, we define $L$ as in \eqref{eq:LQPQ} in which
$\mathcal{Q}$ contains the span of COIN template spectra:
$\mathcal{Q}_1=\span\{P_k\}_{k=1}^{8}$. However, since a single set of templates may not
faithfully represent signal in subsequent experiments (with new measurement errors,
background and baseline noise etc), we enlarge $\mathcal{Q}$ by adding additional
structure related to these templates. For this, set $\mathcal{Q}_2=\span\{P'_k,
P''_k\}_{k=1}^{8}$, where $P'_k$ denotes the derivative of spectrum $P_k$. (Note, to form
$\mathcal{Q}_2$, scale-based approximations to these derivatives are used since raw
differencing of non-smooth spectra introduces noise.) Then set $\mathcal{Q} =
\span\{\mathcal{Q}_1\cup\mathcal{Q}_2\}$ and define
$L=a(I-P_\mathcal{Q})+bP_\mathcal{Q}$.

The regularization parameters in the PEER and ridge estimation processes were chosen
using REML, as described in Section~\ref{sec:TuningParam}.  For the FPCR$_R$ estimate, we
used the \proglang{R}-package \texttt{refund} \cite{ReissRrefund:10} as implemented in
\cite{ReissOgden:07}.

Since $\beta$ is not known (the model $y=X\beta+\epsilon$ is only approximate), we cannot
report MSEs for these three methods.  However, the structure of $P_1$ is qualitatively
known and by experimental design, $y$ is directly associated with $P_1$. The goal here is
that of extracting structure of the constituent spectral components as manifest in a
linear model.  This application is similar to the classic problem of multivariate
calibration \cite{Brown:93,MarxEilersCalib:02} which essentially leads to a regression
model using an experimentally-designed set of spectra from laboratory mixes.

The structure in the estimate here is expected to reflect the structure in $P_1$ that is
correlated with $P_1$'s concentrations, $y$.  The estimate is not, however, expected to
precisely reconstruct $P_1$ since $P_1$ shares structure with the other COIN spectra not
associated with $y$.  See Figure~\ref{Fig:COINs} where $P_1$ is plotted alongside the
other COIN spectra. Now, Figure~\ref{Fig:AOH_BetaHats} shows plots of the PEER, FPCR$_R$
and ridge estimates of the fLM coefficient function.  The PEER estimate,
$\tilde{\beta}_Q$, provides an interpretable compromise between ridge, which involves no
smoothing, and FPCR$_R$, which appears to oversmooth.  For reference, the $P_1$ spectrum
is also plotted along with a mean-adjusted version of $\tilde{\beta}_Q$,
$\tilde{\beta}_Q+\mu$ (dashed line), where $\mu(t) = (1/36)\sum_i x_i(t)$,
$t\in[400,1800]$.

\begin{figure}[h!bt]
\centering
{\includegraphics[height=6cm,width=10cm]{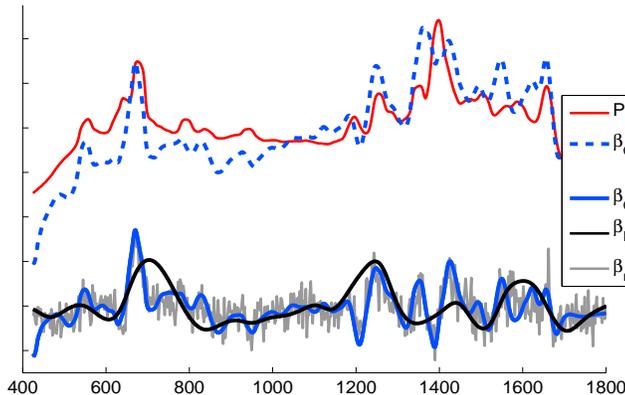}}
\vskip-10pt\caption{\small  Three estimates for a coefficient function that relates concentrations
of $P_1$  to its signal in 8-COIN laboratory mixtures.  Estimates shown: ridge
($\tilde{\beta}_{\mbox{ridge}}$),
FPCR$_R$ ($\tilde{\beta}_F$) and PEER ($\tilde{\beta}_Q$).  For perspective, $P_1$ is plotted (in red) and
the mean-adjusted PEER estimate, $\tilde{\beta}_Q+\mu$ (dashed blue); $\mu$ is the mean
of the mixture spectra $\{x_i\}_{i=1}^{36}$  (not shown).}
\label{Fig:AOH_BetaHats}
\end{figure}

Finally, we consider prediction for these methods by forming a new set of spectra from
different mixture compositions (different concentrations of each COIN) and, additionally,
taken from different batches.  This ``test" set consists of spectra from four technical
replicates in each of 15 mixtures forming a set of $n=60$ spectra,
$\{x_i^{\mbox{\footnotesize test}}\}_{i=1}^n$. As before, $P_1$ is the COIN whose
concentration within each mixture defines the values $\{y_i^{\mbox{\footnotesize
test}}\}_{i=1}^{n}$. For the estimates from each of the three methods (shown in
Figure~\ref{Fig:AOH_BetaHats}) we compute the prediction error:
 $(1/n)\sum_i (y_i^{\mbox{\footnotesize test}}-\langle x_i^{\mbox{\footnotesize
test}},\tilde\beta \rangle)^2$.  The errors for PEER, ridge, and FPCR$_R$
estimates are 0.770,  0.752, 2.139, respectively.  The ridge estimate here
illustrates how low prediction error is not necessarily accompanied by
interpretable structure in the estimate (or low MSE) \cite{CaiHall:06}.

\section{Discussion}
As high-dimensional regression problems become more common, methods that exploit a priori
information are increasingly popular. In this regard, many approaches to penalized
regression are now founded on the idea of ``structured" penalties which impose
constraints based on prior knowledge about the problem's scientific setting. There are
many ways in which such constraints may be imposed, and we have focused on the algebraic
aspects of a penalization process that imposes spatial structure directly into a
regularized estimation.  This approach fits into the classic framework of $L^2$-penalized
regression but with an emphasis on the algebraic role that a penalty operator plays to
impart structure on the estimate.

The interplay between a structured regularization term and the coefficient-function
estimate may not be well understood 
in part because it is not typically viewed in terms of the generalized singular
vectors/values, which is fundamental to this investigation. In particular, any penalized
estimate of the form \eqref{eq:PenEst} with $L\neq I$ is intrinsically based on GSVD
factors in the same way that many common regression methods (such as PCR, ridge,
James-Stein, or partial least squares) are intrinsically based on SVD factors.  Just as
the basics of the ubiquitous SVD are important to understanding these methods, we have
aspired to established the basics of the GSVD as it applies to a this general penalized
regression setting and to illustrate how the theory underlying this approach can be used
inform the choice of penalty operator.

Toward this goal the presentation emphasizes the transparency provided by the
partially-empirical eigenvector expansion \eqref{eq:betaGSVE}.  Properties of the
estimate's variance and bias are determined explicitly by the generalized singular
vectors whose structure is determined by the penalty operator.  We have restricted
attention to additive constraints defined by penalty operators on $L^2$ in order to
retain the direct algebraic connection between the eigenproperties of the operator pair
$(X,L)$ and the spatial structure of $\tilde\beta_{\alpha,L}$. Intuitively, the structure
of the penalty's least-dominant singular vectors should be commensurate with the
informative structure of $\beta$.  The actual effect a penalty has on the properties of
the estimate can be quantified in terms of the GSVD vectors/values.

This perspective differs from popular two-stage signal regression methods in
which estimation is either preceded by fitting the predictors to a set of
(external) basis functions or is followed by a step that smooths the estimate
\cite{CarFerSar:03,HasMal:93,MarxEilers:99,RamSil:05,ReissOgden:07}. Instead,
structure (smoothness or otherwise) is imposed directly into the estimation
process. The implementation of a penalty that incorporates structure less
generic than smoothness (or sparseness) requires some qualitative knowledge
about spatial structure that is informative.  Clearly this is not possible in
all situations, but our presentation has focused on how a functional linear
model may provide a rigorous and analytically tractable way to take advantage
of such knowledge when it exists.

\bigskip
{\bf Acknowledgements:}  The authors wish to thank the referees and associate editor for
all of their questions and helpful suggestions which greatly improved the presentation.
We also gratefully acknowledge the patience, insight and assistance from Dr.~C.~Schachaf
who provided the Raman spectroscopy data for the examples in Sections~\ref{sec:RamanSim}
and \ref{sec:RamanApplic}.  Research support was provided by the National Institutes of
Health grants R01-CA126205, P01-CA053996 and U01-CA086368.

\newpage
\section{Appendix}

\subsection{Defining $\beta$ for the simulation in section~\ref{sec:RamanSim}}\label{App:RamanSimBeta}

This simulation is motivated by an interest in constructing a plausibly realistic $\beta$
whose structure is naturally derived by the scientific setting involving Raman signatures
of nanoparticles.  Although one could model a $\beta$ mathematically using, say,
polynomials or Gaussian bumps (cf., Appendix A.2), such a simulation would be detached
from the physical nature of this problem. Instead, we construct a coefficient function
that genuinely comes from a functional linear model with Raman spectra as predictors.

We first generate in-silico mixtures of COIN spectra as $x^{o}_i = \sum_{k=1}^9 c_{i,k}
P_k$, $i=1,...,50$, where $c_{i,k}\sim \mbox{unif}[0,1]$.
Designating $P_9$ as the COIN of interest, we define response values that correspond to
the ``concentration" of $P_9$ by setting $y^{o}_i:=3\,c_{i,9}$, $i=1,...,n$.  The factor
of 3 imposes a strong association between $P_9$ and the response.

Now, the example in section~\ref{sec:RamanSim} aims to estimate a coefficient function
that truly comes from a solution to a linear model.  However, the equation $y^{o} =
X^{o}\beta$ has infinitely many solutions (where $X^{o}$ is the matrix whose $i$th row is
$x^{o}_i$), so we must we must regularize the problem to obtain a specific $\beta$. For
this, we simply use a ridge penalty and designate the resulting solution to be $\beta$.
This is shown by the dotted curve in Figure~\ref{Fig:COINs} and is qualitatively similar
to $P_9$.

We note that the simulation in section~\ref{sec:RamanSim} uses the same set of COINs, but
a new set of in-silico mixture spectra (i.e., a new set of
$\{c_{i,k}\}\sim\mbox{Unif}[0,1]$).  In addition, a small amount of measurement error was
added, as in section~\ref{sec:BumpsSim}, to each spectrum during the simulation.


\subsection{Frequency domain simulation}\label{App:HallHorSim}

We display results from a study that mimics the scenario of simulations studied by Hall
and Horowitz \cite{HallHor:07}.   We illustrate, in particular, properties of the MSE
discussed following equation~\eqref{eq:MSE_beta_abV} in section~\ref{sec:MSEcompare}
relating to $b=0$.  In fact, we consider the more general scenario in which $\mathcal{Q}$
is not constructed from empirical eigenvectors (as in PCR and ridge), but is defined by a
prespecified envelope of frequencies.

In this simulation both $\beta$ and $x_i$, $i= 1,\ldots,n$, are generated as
sums of the cosine functions
\begin{eqnarray*}
x_i(t)  =  \sum_{j=1}^{40} \gamma_j Z_{ij} \phi_j(t) + e_i(t), \quad
\beta(t)  =  0.75 \phi_5(t) + 1.5 \phi_{11}(t) + 1 \phi_{17}(t), \quad t\in[0,1],
\end{eqnarray*}
where $\gamma_j = (-1)^{j+1} j^{-0.75}$, $Z_{ij}$ is uniformly distributed on
$[-3^{1/2}, 3^{1/2}]$ ($E(Z_{ij})=0$ and $\var(Z_{ij}) = 1$), $\phi_1 \equiv 1$
and $\phi_{j}(t) = 2^{1/2} \cos(j \pi t)$
    for $j \geq 1$, and
$e_i(t) \sim N(0, \sigma_X^2)$, and $\cov(e_i(t), e_{i'}(t'))= 0$ for either $i
\neq i'$ or $t \neq t'$. The response $y_i$ is defined as $y_i = \langle \beta,
x_i\rangle + \epsilon_i$, where $\epsilon_i \sim N(0,\sigma^2), \;
\mbox{i.i.d.}$.  The simulations involve discretizations of these curves
evaluated at $p=100$ equally spaced time points, $t_j$, $j=1, \ldots, p$, that
are common to all curves.

\smallskip \noindent {\bf Penalties. } We consider properties of estimates from
a variety of penalties: ridge ($L=I$), $\mathcal{D}^2$, $\mathcal{D}^2 + a I$,
and PCR\footnote{PCR is not obtained explicitly from a penalty operator, but
see Corollary~\ref{cor:beta_ab_family}.}. In addition, targeted penalties of
the form $L=I-P_\mathcal{Q}$, are defined by the specified subspaces
$\mathcal{Q}=\span\{\phi_j\}_{j\in J}$, for $\phi_j$ defined above.
Specifically, we use $J = J_F = \{j=5,...,17\}$ (a tight envelope of
frequencies) to define $L_F$, and $J = J_G = \{j=4,...,20\}$ (a less focused
span of frequencies) to define $L_G$. The operator $\mathcal{D}^2 + a I$ simply
serves to illustrate the role of higher-frequency singular vectors as discussed
in Section~\ref{sec:GoutisPen}.  In the simulations, the coefficient $a$ in
$\mathcal{D}^2 + a I$ was chosen  simultaneously with $\alpha$ via a
two-dimensional grid search.

\noindent{\bf Simulation results.}  Table~\ref{tab:estF} summarizes estimation
results for all six penalties and two sample sizes, $n=50, 200$. The prediction
results for these estimates are in Table~\ref{tab:predF}. These are reported
for $S/N=10,5$ and $R^2=0.8,0.6$. The number of terms in the PCR estimate was
optimized and ranged from 19 to 25 when $R^2=0.8$ and decreased with decreasing
$R^2$.  Analogously, one could optimize over the dimension of $\mathcal{Q}$ (to
implement a truncated GSVD), but the purpose here is illustrative while in
practice a more robust approach would emply a penalty of the form
\eqref{eq:LQPQ}.

\begin{table}[htbp]
\begin{center}
\caption[]{\small Estimation errors (MSE) for the simulation with localized
frequencies. \\
 } \label{tab:estF}
  \begin{tabular}{rrr||rr|rrrr} \hline
\multicolumn{1}{c}{$n$}&
\multicolumn{1}{c}{$R^2$}&
\multicolumn{1}{r||}{$S/N$}&
\multicolumn{1}{c}{$L_F$}&
\multicolumn{1}{c}{$L_G$}&
\multicolumn{1}{c}{\small PCR} &
\multicolumn{1}{c}{\small ridge} &
\multicolumn{1}{c}{$\mathcal{D}^2$}&
\multicolumn{1}{c}{$\mathcal{D}^2+a\,I$}  \\ \hline
50  & 0.8 & 10  &  42.42 &  77.31 & 123.60 & 132.50 & 1051.20 &  568.45  \\
50  & 0.8 &  5  &  41.55 &  75.41 & 128.75 & 143.48 & 447.64 & 184.07    \\
200 & 0.8 &  10 &  8.28  &  13.48 & 33.44  & 65.78  & 453.54  & 169.41 \\
200 & 0.8 & 5   &  8.56  & 13.08  & 36.36  &  87.59 & 100.15 &  76.24   \\
50  & 0.6 & 10  & 106.89 & 200.08 & 173.56 & 173.94 & 1098.20 & 631.05    \\
50  & 0.6 & 5   & 109.51 & 178.05 & 178.15 & 196.62 & 612.12 & 259.92   \\
200 & 0.6 & 10  &  25.30 &  38.73 &  58.90 &  98.13 & 847.59 &  350.46  \\
200 & 0.6 & 5   &  22.08 &  33.79 &  59.92 & 119.48 & 240.09 & 127.52  \\
\end{tabular}
\end{center}
\end{table}

Errors obtained with ridge and $PCR$ are small, as expected, since the
structure of $\beta$ in this example is consistent with the structure
represented in the singular vectors, $v_k$. Therefore, even though the
relationship between the $y_i$ and $x_i$ degrades (indeed, even as $R^2\to 0$),
these estimates are comprised of vectors that generally capture structure in
$\beta$ since it is strongly represented by the dominant eigenstructure of $X$.
The second-derivative penalty, $\mathcal{D}^2$, produces the worst estimate in
each of the scenarios due to oversmoothing.  Note $\mathcal{D}^2+a\,I$ improves
on $\mathcal{D}^2$, yet it is still not optimal for the range of frequencies in
$\beta$.

Regarding $L_G$, the MSE gets worse as $S/N$ increases.  Indeed, here
$\mathcal{Q}$ is fixed and relatively large and since the $\sigma_k$ decay
faster when $S/N$ is big, this leads to rank deficiency and large variance; see
equation~\eqref{eq:MSE_beta_abV} (note, this only applies approximately since
$\mathcal{Q}$ does not consist of ordinary SVs). In our previous examples, this
is stabilized by a $b>0$.

The problems of estimation and prediction have different properties \cite{CaiHall:06};
good prediction may be obtained even with a poor estimate, as seen in
Table~\ref{tab:predF}. The estimate from $L_{D_a}$ is generally poor relative to others
(as measured by the $L^2$-norm), but its prediction error is comparable with other
methods and is best among the non-targeted penalization methods.  This is consistent with
the outcome described by C.~Goutis \cite{Goutis:98}  where (derivatives of) the predictor
curves contain sharp features and so standard least-squares regularization (OLS, PCR,
ridge, etc.) perform worse than a PEER estimate which imposes a greater emphasis on
``regularly oscillatory but not smooth components"; see section~\ref{sec:GoutisPen}.

\begin{table}[hbtp]
\begin{center}
\caption[]{\small Prediction errors for the simulation with localized frequencies.\\
} \label{tab:predF}
  \begin{tabular}{rrr||rr|rrrc} \hline
  \multicolumn{1}{c}{$n$}&
\multicolumn{1}{c}{$R^2$}&
\multicolumn{1}{r||}{$S/N$}&
\multicolumn{1}{c}{$L_F$}&
\multicolumn{1}{c}{$L_G$}&
\multicolumn{1}{c}{\small PCR} &
\multicolumn{1}{c}{\small ridge} &
\multicolumn{1}{c}{$\mathcal{D}^2$}&
\multicolumn{1}{c}{$\mathcal{D}^2+a\,I$} \\ \hline
50 & 0.8  &  10 &  0.848 &   1.134  &  1.490 &  1.423  &   1.292  &  1.246  \\
50 & 0.8  &   5 &  0.840 &   1.124  &  1.427 &  1.390  &   1.304  &  1.222 \\
200 & 0.8 &  10 &  0.200 &   0.273  &  0.432 &  0.497  &   0.444  &  0.418   \\
200 & 0.8 &   5 &  0.211 &   0.276  &  0.460 &  0.547  &   0.466  &  0.455  \\
50  & 0.6 &  10 &  2.165 &   2.900  &  3.051 &  2.705  &   2.621  &  2.472   \\
50  & 0.6 &   5 &  2.171 &   2.832  &  3.158 &  2.938  &   2.912  &  2.724   \\
200 & 0.6 &  10 &  0.621 &   0.801  &  1.058 &  1.160  &   1.044  &  0.990  \\
200 & 0.6 &   5 &  0.584 &   0.766  &  1.062 &  1.186  &   1.069  &  1.014   \\
\end{tabular}
\end{center}
\end{table}

\newpage

\bibliography{PEERBibliography}
\bibliographystyle{amsplain}

\end{document}